\documentclass[11pt]{article}
\usepackage[margin=1in]{geometry}

\usepackage{booktabs} 

\usepackage{xspace}
\usepackage[utf8]{inputenc}
\usepackage{mathtools}
\usepackage{amsthm}
\usepackage{amssymb}
\usepackage{hyperref}
\usepackage{algorithm}
\usepackage{algpseudocode}
\usepackage[toc,page]{appendix}
\usepackage{tikz}
\usepackage{esvect}
\usepackage{xcolor}
\usepackage[T1]{fontenc}
\usepackage[square,numbers]{natbib}
\usepackage{bbold}
\usepackage{caption}

\usetikzlibrary{arrows,shapes,calc}
\usetikzlibrary{decorations.pathreplacing}
\usetikzlibrary{decorations.markings}
\usetikzlibrary{positioning}
\makeatletter

\newtheorem{theorem}{Theorem}
\newtheorem{definition}[theorem]{Definition}
\newtheorem{lemma}[theorem]{Lemma}
\newtheorem{corollary}[theorem]{Corollary}

\newtheorem{conjecture}[theorem]{Conjecture}

\newcommand{\Mod}[1]{\ \mathrm{mod}\ #1}
\newcommand{\floor}[1]{\left\lfloor #1 \right\rfloor}
\newcommand{\ceil}[1]{\left\lceil #1 \right\rceil}

\newcommand{\pnat}{\mathbb{N}_{\ge 0}}

\newcommand{\eps}{\epsilon}

\newcommand{\prob}{\mathbb{P}}
\newcommand{\Oh}{\mathcal{O}}
\newcommand{\Ot}{\mathcal{\widetilde{O}}}
\newcommand{\dxy}{d\big(x[i], y[(k+i)\; (\text{mod}\; 2n)]\big)}
\newcommand{\naivecitemany}[2]{#1~et~al.~\cite{#2}}
\newcommand{\naiveciteone}[2]{#1~\cite{#2}}

\newcommand{\maxconv}{\textsc{MaxConv}\xspace}
\newcommand{\maxconvupper}{\textsc{MaxConv UpperBound}\xspace}
\newcommand{\maxconvlower}{\textsc{MaxConv LowerBound}\xspace}
\newcommand{\minconv}{\textsc{MinConv}\xspace}
\newcommand{\knapsack}{\textsc{0/1 Knapsack}\xspace}
\newcommand{\knapsackplus}{$\textsc{0/1 Knapsack}^{+}$\xspace}
\newcommand{\subsetsum}{\textsc{SubsetSum}\xspace}
\newcommand{\uknapsack}{\textsc{Unbounded Knapsack}\xspace}
\newcommand{\convsum}{\textsc{3sumConv}\xspace}
\newcommand{\necklace}{$l_\infty$-\textsc{Necklace Alignment}\xspace}
\newcommand{\sparsity}{\textsc{Tree Sparsity}\xspace}
\newcommand{\superadditivity}{\textsc{SuperAdditivity Testing}\xspace}
\newcommand{\mcsp}{\textsc{MCSP}\xspace}

\newcommand{\defproblemu}[3]{
  \vspace{2mm}
  \vspace{1mm}
\noindent\fbox{
  \begin{minipage}{0.95\textwidth}
  #1 \\
  {\bf{Input:}} #2  \\
  {\bf{Task:}} #3
  \end{minipage}
  }
  \vspace{2mm}
}

\newcommand{\footremember}[2]{%
    \footnote{#2}
    \newcounter{#1}
    \setcounter{#1}{\value{footnote}}%
}
\newcommand{\footrecall}[1]{%
    \footnotemark[\value{#1}]%
} 

\begin{document}
\title{On problems equivalent to
(min,+)-convolution\footnote{Extended abstract of this paper was presented at ICALP
2017~\cite{icalp-2017}. Some results of this paper have been published
independently~\cite{kunnemann-icalp2017}. This work is part of a project TOTAL that has received funding from the European Research Council (ERC) under the European Union’s Horizon 2020 research and innovation programme (grant agreement No 677651).}}

\author{%
    Marek Cygan\footremember{MIMUW}{Institute of Informatics, University of
    Warsaw, Poland, \texttt{\{cygan, mucha, k.wegrzycki, m.wlodarczyk\}@mimuw.edu.pl}}
    \and
    Marcin Mucha\footrecall{MIMUW}
    \and
    Karol W\k{e}grzycki\footrecall{MIMUW}
    \and
    Micha\l{} W\l{}odarczyk\footrecall{MIMUW}
}
\date{}

\maketitle

\thispagestyle{empty}
\begin{abstract}
In recent years, significant progress has been made in explaining the apparent hardness of
improving upon the naive solutions for many fundamental polynomially solvable
problems. This progress has come in the form of conditional lower bounds -- reductions from a problem assumed to be hard.
The hard problems include 3SUM, All-Pairs Shortest Path, SAT, Orthogonal Vectors, and others. 

In the $(\min,+)$-convolution problem, the goal is to compute a sequence
$(c[i])^{n-1}_{i=0}$, where $c[k] = $ $\min_{i=0,\ldots,k} $ $\{a[i] $ $+$ $b[k-i]\}$, given
sequences $(a[i])^{n-1}_{i=0}$ and 
$(b[i])_{i=0}^{n-1}$. This can easily be done in $\Oh(n^2)$ time, but no $\Oh(n^{2-\varepsilon})$ algorithm
is known for $\varepsilon > 0$. In this paper, we undertake a systematic study of the $(\min,+)$-convolution problem
as a hardness assumption. 

First, we establish the equivalence of this problem to a group of other problems, including variants of the classic
knapsack problem and problems related to subadditive sequences. The
$(\min,+)$-convolution problem has been used
as a building block in algorithms for many problems, notably problems in stringology. It has also appeared as an
ad hoc hardness assumption. Second, we investigate some of these connections and provide new reductions and other results.
We also explain why replacing this assumption with the SETH might not be possible for some problems.

\end{abstract}

\clearpage
\setcounter{page}{1}

\section{Introduction}

\subsection{Hardness in P}

For many problems, there exist ingenious algorithms that significantly improve upon the
naive approach in terms of time complexity. On the other hand, for some fundamental
problems, the naive algorithms are still the best known or have been improved upon only
slightly. To some extent, this has been explained by the P$\ne$NP conjecture. However, 
for many problems, even the naive approaches lead to polynomial algorithms, and
the P$\ne$NP conjecture does not seem to be particularly useful for proving polynomial lower bounds.

In recent years, significant progress has been made in establishing such bounds, 
conditioned on conjectures other than P$\ne$NP. Each conjecture
claims time complexity lower bounds for a different problem. The main
conjectures are as follows. First, the conjecture that there is no $\Oh(n^{2-\epsilon})$ algorithm for the 3SUM
problem\footnote{We present all problem definitions, together with the known results,
    for these problems in Section~\ref{problems}. This is to keep the introduction
relatively free of technicalities.} implies hardness for
problems in computational geometry~\cite{3sum2} and dynamic
algorithms~\cite{3sum1}. Second, the conjecture that
there is no algorithm $\Oh(n^{3-\eps})$ for All-Pairs Shortest Path (APSP) implies the hardness of determining the graph
radius and graph median and the hardness of some dynamic problems (see~\cite{conjectures-survey}
for a survey of related results). Finally, the Strong Exponential Time Hypothesis (SETH) introduced in~\cite{seth1,seth2}
has been used extensively to prove the hardness of parametrized
problems~\cite{seth-treewidth,fpt-book} and has recently
lead to polynomial lower bounds via the intermediate Orthogonal Vectors problem (see~\cite{ov-conjecture}). 
These include bounds for the Edit Distance~\cite{edit-distance}, Longest Common
Subsequence~\cite{lcs1,lcs2}, and others~\cite{conjectures-survey}.

It is worth noting that in many cases, the results mentioned indicate
not only the hardness of the problem in question but also that it is computationally
equivalent to the underlying hard problem. This leads to clusters of equivalent
problems being formed, each cluster corresponding to a single hardness
assumption (see~\cite[Figure 1]{conjectures-survey}).

As Christos H. Papadimitriou stated, ``\emph{There is nothing wrong with trying
to prove that P=NP by developing a polynomial-time algorithm for an NP-complete
problem. The point is that without an NP-completeness proof we would be trying
the same thing without knowing it!}''~\cite{cytat}. In the same spirit, 
these new conditional hardness results have cleared the polynomial
landscape by showing that there really are not that many hard
problems (for the recent background, see~\cite{williamsICM}).

\subsection{Hardness of \minconv}
\label{sec:contribution}

In this paper, we propose yet another hardness assumption: the \minconv problem.

\defproblemu{\minconv}
{Sequences $(a[i])_{i=0}^{n-1},\, (b[i])_{i=0}^{n-1}$}
{Output sequence $(c[i])_{i=0}^{n-1}$, such that $c[k] = \min_{i+j=k} (a[i]+b[j])$}

This problem has previously been used as a hardness assumption for at least two
specific problems~\cite{mscp,sparsity}, but to the best of our knowledge, no attempts have been made to systematically study the
neighborhood of this problem in the polynomial complexity landscape.

 To be more precise, in all problem definitions, we assume that the input sequences consist of integers in the range $[-W,W]$.
 Following the design of the APSP conjecture~\cite{apsp2},
 we allow $\text{polylog}(W)$ factors in the definition of a subquadratic running time.

\begin{conjecture}\label{conj:minconv}
    There is no $\Oh(n^{2-\varepsilon}\mathrm{polylog}(W))$ algorithm for \minconv when $\varepsilon>0$.
\end{conjecture}

Let us first look at the place occupied by \minconv in the landscape
of established hardness conjectures. Figure~\ref{fig:conjectures} shows known reductions between these conjectures and includes
\minconv. \naivecitemany{Bremner}{necklaces} showed the reduction from \minconv  to
APSP. It is also known~\cite{sparsity} that \minconv can be
reduced to 3SUM by using reductions~\cite{3sum1} and \cite[Proposition 3.4,
Theorem 3.3]{weighted-subgraphs} (we provide the details in the Appendix~\ref{sec:3sum}). Note
that a reduction from 3SUM or APSP to \minconv would imply a reduction between 3SUM and APSP, which is a major open problem in this area of study~\cite{conjectures-survey}. No relation between \minconv and SETH or OV is known. 

\begin{figure}[ht!]
    \centering
    \begin{tikzpicture}
    \begin{scope}[every node/.style={rectangle,thick,draw}]
        \node (APSP) at (0,0) {APSP};
        \node (OV) at (3,0) {OV};
        \node (SETH) at (3,1.5) {SETH};
        \node (SUM) at (6,0) {3SUM};
        \node (MINCONV) at (3,-2) {\minconv};
    \end{scope}

    \begin{scope}[every node/.style={},
                  every edge/.style={draw=black,very thick}]
        \path [->] (SETH) edge node[left] {\cite{ov-conjecture}} (OV);
        \path [->] (OV) edge node {\scalebox{2}{$\color{red} \times$}} node[yshift=0.5cm] {\cite{nondet-seth}} (SUM);
        \path [->] (OV) edge node {\scalebox{2}{$\color{red} \times$}} node[yshift=0.5cm] {\cite{nondet-seth}} (APSP);
        \path [->] (MINCONV) edge[bend left=10] node[below] {\cite{necklaces}} (APSP);
        \path [->] (MINCONV) edge[bend right=10] node[below, right]
        {\cite{sparsity,3sum1,weighted-subgraphs}} (SUM);
    \end{scope}
    \end{tikzpicture}
    \vspace{-5pt}

    \caption{The relationship between popular conjectures. A reduction from OV
        to 3SUM or APSP contradicts the nondeterministic version of SETH~\cite{nondet-seth, conjectures-survey}
        (these arrows are striked out).}
    \label{fig:conjectures}
\end{figure}
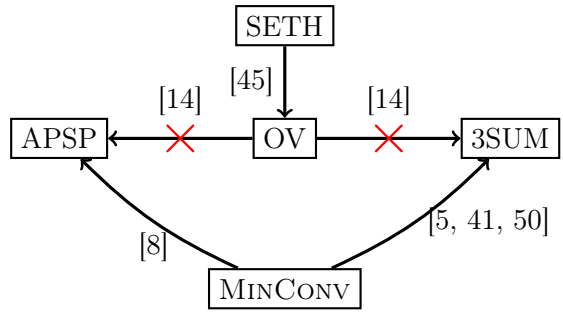

In this paper, we study three broad categories of problems. The
first category consists of the classic \knapsack and its variants, which
we show to be essentially equivalent to \minconv. This is perhaps somewhat
surprising given the recent progress of Bringmann~\cite{bringmann} for \subsetsum,
which is a special case of \knapsack. However, note that Bringmann's
algorithm~\cite{bringmann} (as well as other efficient solutions for
\subsetsum) is built upon the idea of composing solutions using the
$(\vee,\wedge)$-convolution, which can be implemented efficiently using a Fast
Fourier Transform (FFT). The corresponding composition operation for
\knapsack is \minconv (see Section~\ref{knapsack-bringmann} for details). 

The second category consists of problems directly related to \minconv. This includes 
decision versions of \minconv and problems related to the notion of subadditivity. Any subadditive sequence
$a$ with $a[0] = 0$ is an idempotent of \minconv; thus, it is perhaps unsurprising that these problems are equivalent to \minconv.

Finally, we investigate problems that have previously been shown to be related to \minconv and then contribute 
some new reductions or simplify existing ones. Moreover, some of the results of
this paper have been published independently by \citet{kunnemann-icalp2017} at
the same conference.

\section{Problem definitions and known results}
\label{problems}

\subsection{3SUM}

\defproblemu{\textsc{3sum}}
{Sets of integers $A,B,C$, each of size $n$}
{Decide whether there exist $a\in A,\,b\in B,\,c\in C$ such that $a+b=c$}

The \textsc{3sum} problem is the first problem that was considered
as a hardness assumption in~P.
It admits a simple $\Oh(n^2\log{n})$ algorithm, but the existence of an $\Oh(n^{2-\epsilon})$ algorithm
remains a major open problem.
The first lower bounds based on the hardness of \textsc{3sum} appeared in 1995~\cite{3sum2};
some other examples can be found in~\cite{3sum-best, 3sum1, weighted-subgraphs}.
The current best algorithm for \textsc{3sum} runs in slightly subquadratic
expected time $\Oh\left((n^2/\log^2{n})(\log\log{n})^2\right)$~\cite{3sum-best}.
An~$\Oh\left(n^{1.5}\mathrm{polylog}(n)\right)$
algorithm is possible on a nondeterministic
Turing machine~\cite{nondet-seth} (see Section~\ref{prelim} for the definition of
nondeterministic algorithms).

The \textsc{3sum} problem is known to be subquadratically equivalent to its
convolution version in a randomized setting~\cite{3sum1}.

\defproblemu{\convsum}
{Sequences $a,b,c$ of integers, each of length $n$}
{Decide whether there exist $i,j$ such that $a[i]+b[j]=c[i+j]$}

\noindent Both problems are sometimes considered with real weights, but in this work,
we restrict them to only the integer setting.

\subsection{\minconv}
\label{minconv-algorithms}

We have already defined the \minconv problem in Subsection~\ref{sec:contribution}.
Note that it is equivalent (just by negating elements) to the analogous \maxconv problem.

\defproblemu{\maxconv}
{Sequences $(a[i])_{i=0}^{n-1},\, (b[i])_{i=0}^{n-1}$}
{Output sequence $(c[i])_{i=0}^{n-1}$, such that $c[k] = \max_{i+j=k} (a[i]+b[j])$}

\noindent We describe our contribution in terms of \minconv, as this version has been already been heavily studied.
However, in the theorems and proofs, we use \maxconv, as it is easier to work with.
We also work with a decision version of the problem. Herein, we will use
$a \oplus^{\max} b$ to denote the \maxconv of two sequences $a$ and $b$.

\defproblemu{\maxconvupper}
{Sequences $(a[i])_{i=0}^{n-1},\, (b[i])_{i=0}^{n-1}, \,(c[i])_{i=0}^{n-1}$}
{Decide whether $c[k] \ge \max_{i+j=k} (a[i]+b[j])$ for all $k$}

If we replace the latter condition with $c[k] \le \max_{i+j=k} (a[i]+b[j])$,
we obtain a similar problem \maxconvlower.
Yet another statement of a decision version asks whether a given sequence
is a self upper bound with respect to \maxconv, i.e., if it is superadditive.
From the perspective of \minconv, we may ask an analogous question about being subadditive
(again, equivalent by negating elements).
As far as we know, the computational complexity of these problems has not yet been studied.

\defproblemu{\superadditivity}
{A sequence $(a[i])_{i=0}^{n-1}$}
{Decide whether $a[k] \ge \max_{i+j=k} (a[i]+a[j])$ for all $k$}

In the standard $(+,\cdot)$ ring, convolution can be computed in $\Oh(n\log{n})$ time
by the FFT.
A natural way to approach \minconv would be to design an analogue of FFT in the $(\min,+)$-semi-ring,
also called the \emph{tropical semi-ring}\footnote{In this setting, \minconv is often called a $(\min,+)$-convolution, inf-convolution or epigraphic sum.}. 
However, due to the lack of an inverse for the $\min$-operation, it is unclear if such a transform exists for general sequences.
When restricted to convex sequences, one can use a tropical analogue of FFT, namely, the
Legendre-Fenchel transform~\cite{fenchel}, which can be performed in linear time~\cite{legendre-linear}.
\cite{3sum-reporiting-esa} also considered sparse variants of convolutions and
their connection with \textsc{3sum}.

There has been a long line of research dedicated to improving upon the $\Oh(n^2)$ algorithm for
\minconv.
\naivecitemany{Bremner}{necklaces} presented an $\Oh(n^2/\log{n})$ algorithm for \minconv,
as well as a reduction from \minconv to APSP~\cite[Theorem 13]{necklaces}.
\naiveciteone{Williams}{apsp-circuit} developed an $\Oh({n^3}/{2^{\Omega(\log{n})^{1/2}}})$
algorithm for APSP, which can also be used to obtain an
$\Oh(n^2/2^{\Omega(\log{n})^{1/2}})$ algorithm for \minconv~\cite{clustered-3sum}.

Truly subquadratic algorithms for \minconv exist for monotone increasing
sequences with integer values bounded by $\Oh(n)$. \naiveciteone{Chan and
Lewenstein}{clustered-3sum} presented an~$\Oh(n^{1.859})$
randomized algorithm and an~$\Oh(n^{1.864})$ deterministic algorithm for this
case. They exploited ideas from additive combinatorics.
\naivecitemany{Bussieck}{bussieck} showed that for a random input, \minconv can be
computed in $\Oh(n\log{n})$ expected and $\Oh(n^2)$ worst-case time.

If we are satisfied with computing $c$ with a relative error $(1+\eps)$,
then the general \minconv admits a nearly linear algorithm~\cite{sparsity, zwick}.
It could be called an FPTAS (fully polynomial-time approximation scheme),
noting that this name is usually reserved for single-output problems
for which decision versions are NP-hard. 

Using the techniques of \naivecitemany{Carmosino}{nondet-seth} and the reduction from \maxconvupper to \textsc{3sum}
(see Appendix~\ref{sec:3sum}), one can construct an
$\Oh\left(n^{1.5}\mathrm{polylog}(n)\right)$
algorithm that works on nondeterministic Turing machines for \maxconvupper (see
Lemma~\ref{cor:maxconv-nondet}).
This running time matches the $\Oh(n^{1.5})$ algorithm for \minconv in the nonuniform decision tree
model~\cite{necklaces}.
This result is based on the techniques of~Fredman~\cite{fredman,fredman2}.
It remains unclear how to transfer these results to the word-RAM model~\cite{necklaces}.

\subsection{\textsc{Knapsack}}

\defproblemu{\knapsack}
{A set of items $\mathcal{I}$ with given integer weights and values $\left((w_i,v_i)\right)_{i \in \mathcal{I}}$, capacity $t$}
{Find the maximum total value of a subset $\mathcal{I'} \subseteq \mathcal{I}$ such that $\sum_{i \in \mathcal{I'}}w_i \le t$}

If we are allowed to take multiple copies of a single item, then we obtain the \uknapsack problem.
The decision versions of both problems are known to be NP-hard~\cite{garey1979computers}, but
there are classical algorithms based on dynamic programming with a pseudopolynomial
running time $\Oh(nt)$~\cite{bellman}.

In fact, they are used to solve more general problems, i.e., $\textsc{0/1 Knapsack}^{+}$ and $\textsc{Unbounded Knapsack}^{+}$,
where we are asked to output answers for each $0 < t' \le t$.
There is also a long line of research on FPTAS for \textsc{Knapsack}, with the current best running times being
$\Ot(n+\frac{1}{\eps^{2.4}})$ for \knapsack~\cite{chan-sosa}
and $\Ot(n + \frac{1}{\eps^2})$ for \uknapsack~\cite{uknapsack-ptas}.

\subsection{Other problems related to \minconv}
\label{sec:otherproblems}

\defproblemu{\sparsity}
{A rooted tree $T$ with a weight function $w:V(T) \rightarrow\pnat$, parameter $k$}
{Find the maximum total weight of a rooted subtree of size $k$}

The \sparsity problem admits an $\Oh(nk)$ algorithm, which was at first invented for
the restricted case of balanced trees~\cite{cartis2013exact} and then later generalized~\cite{sparsity}.
There is also a nearly linear FPTAS
based on the FPTAS for \minconv~\cite{sparsity}.
It is known that an $\Oh(n^{2-\eps})$ algorithm for \sparsity entails
a subquadratic algorithm for \minconv~\cite{sparsity}.

\defproblemu{\textsc{Maximum Consecutive Subsums Problem} (\mcsp)}
{A sequence $(a[i])_{i=0}^{n-1}$}
{Output the maximum sum of $k$ consecutive elements for each $k$}

There is a trivial $\Oh(n^2)$ algorithm for \mcsp
and a nearly linear FPTAS based on the FPTAS for \minconv~\cite{approximate-mscp}.
To the best of our knowledge, this is the first problem to have been
explicitly proven to be subquadratically equivalent to \textsc{MinConv}~\cite{mscp}.
Our reduction to \superadditivity allows us to significantly simplify the proof
(see Section~\ref{sec:mcsp}).

\defproblemu{$l_p$-\textsc{Necklace Alignment}}
{Sequences $(x[i])_{i=0}^{n-1},\, (y[i])_{i=0}^{n-1} \in [0,1)^n$ describing locations of beads on a circle}
{Output the cost of the best alignment in the $p$-norm, i.e., $\sum_{i=0}^{n-1}
d\left(x[i] + c, y[i+s \; (\text{mod} \; n)]\right)^p$,
where $c \in [0,1)$ is a circular offset, $s \in \{0,\ldots,n-1\}$ is a shift, and $d$ is a distance function on a circle}

In the $l_p$-\textsc{Necklace Alignment} problem, we are given two sorted sequences
of real numbers $(x[i])_{i=0}^{n-1}$ and $(y[i])_{i=0}^{n-1}$ that represent
two necklaces. We assume that each number in the sequence represents a point on a
circle (we refer to this circle as the \emph{necklace} and the points on it as the
\emph{beads}). The distance between beads
$x_i$ and $y_j$ is defined in~\cite{necklaces} as:

\begin{displaymath}
    d(x_i,y_j) = \min \{ |x_i - y_j|, (1- |x_i - y_j|) \}
\end{displaymath}

to represent the minimum between the clockwise and counterclockwise distances
along the circular necklaces. The $l_p$-\textsc{Necklace Alignment} is an
optimization problem where we can manipulate two parameters. The first parameter
is the offset $c$, which is the clockwise rotation of the necklace
$(x[i])_{i=0}^{n-1}$ relative to the necklace $(y[i])_{i=0}^{n-1}$.
The second parameter is the shift $s$, which defines the perfect matching
between beads from the first and second necklaces, i.e., bead $x[i]$ matches
bead $y[i+s \; (\text{mod} \; n)]$ (see~\cite{necklaces}).

For $p=\infty$, we are interested in bounding the maximum distance between any two matched beads.
The problem initially emerged for $p=1$ during research on the geometry of musical rhythm~\cite{music}.
The family of \textsc{Necklace Alignment} problems was systematically studied by~\naivecitemany{Bremner}{necklaces}
for various values of $p$.
For $p=2$, they presented an $\Oh(n\log n)$ algorithm based on the FFT.
For $p=\infty$, the problem was reduced to \minconv, which led to a slightly subquadratic algorithm.
This makes $l_\infty$-\textsc{Necklace Alignment} a natural problem to study in the context of
\minconv-based hardness. Interestingly, we are not able to show such hardness, which
presents an intriguing open problem. Instead we reduce $l_\infty$-\textsc{Necklace Alignment} to a related problem. 

Although it is more natural to state the problem with inputs from $[0,1)$,
we find it more convenient to work with integer sequences that describe a necklace after scaling.


Fast $o(n^2)$ algorithms for \minconv have also found applications in text algorithms.
\naiveciteone{Moosa and Rahman}{strings1} reduced \emph{Indexed Permutation Matching}
to \minconv and obtained an $o(n^2)$ algorithm. \naivecitemany{Burcsi}{strings3}
used \minconv to obtain faster algorithms for \emph{Jumbled Pattern
Matching} and described how finding dominating pairs can be used to solve
\minconv. Later, \naivecitemany{Burcsi}{strings2} showed that fast \minconv can also be used to obtain faster
algorithms for a decision version of \emph{Approximate Jumbled Pattern
Matching} over binary alphabets.

\section{Summary of new results}

Figure~\ref{fig:our-conjectures} illustrates the technical contributions of this paper.
The long ring of reductions on the left side of the Figure~\ref{fig:our-conjectures} is summarized below.

\begin{theorem}
    \label{thm:main}
    The following statements are equivalent:
    \begin{enumerate}
        \item There exists an $\Oh(n^{2-\varepsilon})$ algorithm for \maxconv for some $\varepsilon>0$.
        \item There exists an $\Oh(n^{2-\varepsilon})$ algorithm for \maxconvupper for some $\varepsilon>0$.
        \item There exists an $\Oh(n^{2-\varepsilon})$ algorithm for \superadditivity for some $\varepsilon>0$.
        \item There exists an $\Oh((n+t)^{2-\varepsilon})$ algorithm for \uknapsack for some $\varepsilon>0$.
        \item There exists an $\Oh((n+t)^{2-\varepsilon})$ algorithm for \knapsack for some $\varepsilon>0$.
    \end{enumerate}
    We allow randomized algorithms.
\end{theorem}



\begin{figure}[ht!]
    \centering
    \begin{tikzpicture}
    \begin{scope}[every node/.style={rectangle,thick,draw}]
        \node (minconv) at (0,0) {\maxconv};
        \node (maxconvupper) at (-3,2) {\maxconvupper};
        \node (superadditivity) at (-6,0.75) {\superadditivity};
        \node (uknapsack) at (-6,-0.75) {\uknapsack};
        \node (knapsack) at (-3,-2) {\knapsack};

        \node (mscp) at (3,0) {\mcsp};
        \node (minconvupper) at (0,3.5) {\maxconvlower};
        \node (sparsity) at (3,-3) {\sparsity};
        \node (necklace) at (2,2) {\necklace};
    \end{scope}

    \begin{scope}[every node/.style={},
                  every edge/.style={draw=red,very thick}]

        \path [->] (minconv)            edge[bend right=10] node[right]      {\ref{maxconv-maxconvupper}} (maxconvupper);
        \path [->] (maxconvupper)       edge[bend right=10] node[left]       {\ref{maxconvupper-superadditivity}} (superadditivity);
        \path [->] (superadditivity)    edge[bend right=10] node[left]       {\ref{superadditivity-uknapsack}} (uknapsack);
        \path [->] (uknapsack)          edge[bend right=10] node[left,below] {\ref{uknapsack-knapsack}} (knapsack);
        \path [->] (knapsack)           edge[bend right=10] node[right]      {\ref{knapsack-maxconv}} (minconv);

        \path [->] (sparsity) edge[bend left=30] node[left]          {\ref{thm:sparsity}} (minconv);
        \path [->] (minconvupper) edge[bend left=10] node[right]    {\ref{thm:necklace}} (necklace);
    \end{scope}
    \begin{scope}[every node/.style={},
                  every edge/.style={draw=black,very thick,dashed}]
        \path [->] (minconv) edge[bend left=10] node[right]          {\cite{sparsity}} (sparsity);
        \path [->] (necklace) edge[bend left=10] node[left]          {\cite{necklaces}} (minconv);

        \path [->] (minconv) edge[bend left=10] node[above]          {\cite{mscp}} (mscp);
        \path [->] (mscp)    edge[bend left=10] node[below]          {\cite{mscp}} (minconv);
    \end{scope}

    \end{tikzpicture}
    \caption{Summary of reductions in the \minconv complexity class. An arrow from problem $A$ to $B$ denotes a reduction from $A$ to $B$. Black dashed arrows were previously known, while red arrows correspond to new results. Numbers next to the red arrows indicate the corresponding theorems. The only randomized reduction is in the proof of Theorem~\ref{knapsack-maxconv}.}
    \label{fig:our-conjectures}
\end{figure}
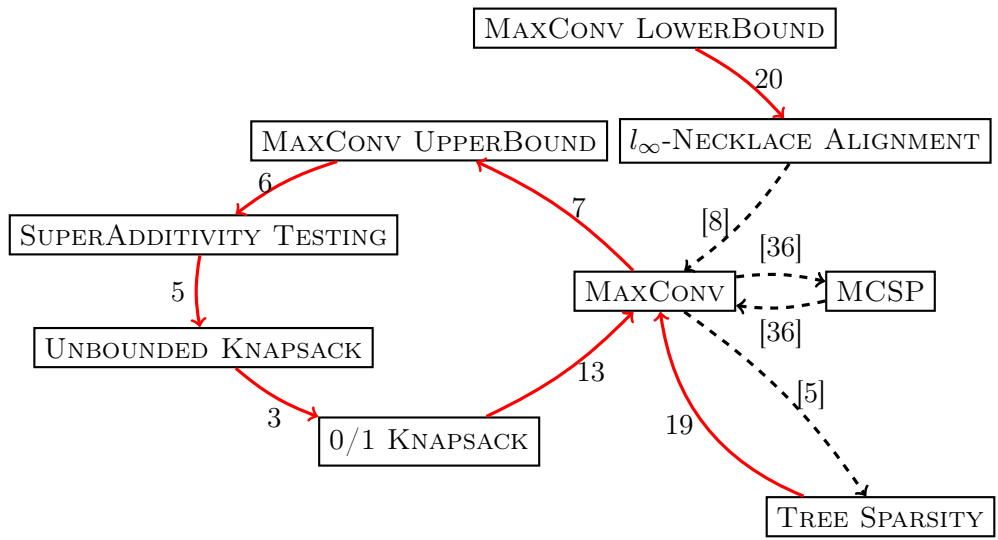



Theorem~\ref{thm:main} is split into five implications, presented separately as 
Theorems~\ref{uknapsack-knapsack},~\ref{superadditivity-uknapsack},~\ref{maxconvupper-superadditivity},~\ref{maxconv-maxconvupper}
and~\ref{knapsack-maxconv}.
While Theorem~\ref{thm:main} has a relatively short and simple statement, it is not the
strongest possible version of the equivalence. In particular, one can show analogous
implications for subpolynomial improvements, such as the
$\Oh(n^2/2^{\Omega(\log n)^{1/2}})$ algorithm for \minconv presented by Williams~\cite{apsp-circuit}.
The theorems listed above contain stronger versions of the implications.
The proof of Theorem~\ref{maxconv-maxconvupper} has
been independently given in~\cite{sparsity}. We present it here
because it is the first step in the ring of reductions and introduces
the essential technique of~\naiveciteone{Vassilevska and Williams}{vassilevska2010subcubic}.

Section~\ref{sec:other} is devoted to the remaining arrows in Figure~\ref{fig:our-conjectures}.
In Subsection~\ref{sec:mcsp}, we show that by using Theorem~\ref{thm:main},
we can obtain an alternative proof of the equivalence of \mcsp and \maxconv (and thus also \minconv), which is much simpler
than that presented in~\cite{mscp}. 
In Subsection~\ref{sec:sparsity}, we show that \sparsity reduces to \maxconv, complementing the
opposite reduction shown in~\cite{sparsity}. 
We also provide some observations on the
possible equivalence between \necklace and \maxconv in Subsection~\ref{sec:necklace}. 

The relation between \maxconv and \textsc{3sum} implies that we should not expect
the new conjecture to follow from SETH. In Section~\ref{sec:nondet}, we exploit the revealed connections between problems
to show that it might also not be possible to replace the hardness assumption for \uknapsack with SETH.
More precisely, we prove that there can be no deterministic reduction from \textsc{SAT} to \uknapsack
that would rule out running time $\Oh(n^{1-\varepsilon}t)$ under the assumption of NSETH.

\section{Preliminaries}
\label{prelim}

We present a series of results of the following form:
if a problem $\mathcal{A}$ admits an algorithm with running time $T(n)$,
then a problem $\mathcal{B}$ admits an algorithm with running time $T'(n)$,
where function $T'$ depends on $T$ and $n$ is the length of the input.
Our main interest is to show that $T(n) = \Oh(n^{2-\eps}) \Rightarrow T'(n) = \Oh(n^{2-\eps'})$.
Some problems, in particular \textsc{Knapsack}, have no simple
parametrization, and we allow function $T$ to take multiple arguments.

In this paper, we follow the convention of~\cite{nondet-seth} and say that the
decision problem $L$ admits a nondeterministic algorithm in time $T(n)$ if $L
\in \text{NTIME}(T(n)) \cap \text{co-NTIME}(T(n))$.

We assume that for all studied problems, the input consists of a list of integers
within $[-W,W]$. 
Since Conjecture~\ref{conj:minconv} is oblivious to $\text{polylog}(W)$ factors, we omit $W$ as a running time
parameter and allow function $T$ to hide factor $\text{polylog}(W)$ for the sake of readability.  
We also use $\Ot$~notation to explicitly hide
polylogarithmic factors with respect to the argument.
Herein, we will use $a\oplus^{\max}b$ to denote the \maxconv of
sequences $a,b$ (see Subsection~\ref{piknapsack}).

As the size of the input may increase during our reductions,
we restrict ourselves to a class of functions satisfying $T(cn) = \Oh(T(n))$ for a constant $c$.
This is justified, as we focus on functions of the form $T(n) = n^\alpha$.
In some reductions, the integers in the new instance may increase to $\Oh(nW)$.
In these cases, we multiply the running time by $\text{polylog}(n)$ to take into account the overhead
of performing arithmetic operations. 
All logarithms are base 2.

\section{Main reductions}
\label{sec:main}

\begin{theorem}[\uknapsack $\rightarrow$ \knapsack]
\label{uknapsack-knapsack}
    A $T(n,t)$ algorithm for \knapsack implies an $\Oh\left(T(n\log{t},t)\right)$ algorithm for \uknapsack.
\end{theorem}
\begin{proof}
Consider an instance of \uknapsack with capacity $t$ and the
set of items given as weight-value pairs $\left((w_i,v_i)\right)_{i \in \mathcal{I}}$.
Construct an equivalent \knapsack instance with the same $t$ and
the set of items $\left((2^jw_i,2^jv_i)\right)_{i \in \mathcal{I}, 0 \le j \le \log{t}}$.
Let $X = (x_i)_{i \in \mathcal{I}}$ be the list of multiplicities of items chosen in a
solution to the \uknapsack problem.
Of course, $x_i \le t$.
Define $(x_i^j)_{0 \le j \le \log{t}},\, x_i^j \in \{0,1\}$ to be the binary representation of $x_i$.
Then, the vector $(x_i^j)_{i \in \mathcal{I}, 0 \le j \le \log{t}}$ induces a solution
to \knapsack with the same total weight and value.
The described mapping can be inverted. This implies the equivalence between the instances
and proves the claim.
\end{proof}

We now consider the \superadditivity problem. We start by showing that we can
consider only the case of nonnegative monotonic sequences. This is a useful,
technical assumption that simplifies the proofs. 

\begin{lemma}
    \label{nonnegative-sequece}
    Every sequence $(a[i])_{i=0,\ldots,n-1}$ can be transformed in linear time
    to a nonnegative monotonic sequence $(a'[i])_{i=0,\ldots,n-1}$ such that $a[i]$ is
    superadditive iff $a'[i]$ is superadditive.
\end{lemma}

\begin{proof}

First, note that if $a[0] > 0$, then the sequence is not superadditive for
$n>0$ because $a[0] + a[i] > a[i]$. In
the case where $a[0] \le 0$, the $0$-th element does not influence the result of
the algorithm. Thus, we can set $a'[0] = 0$ to ensure the nonnegativity of $a'$. Next, to
guarantee monotonicity, we choose $C > 2 \max_i\{|a[i]|\}$.  Let

\begin{displaymath} 
    a'[i] = \begin{cases} 
        0, &\text{if } i = 0\\
        C i + a[i],& \text{otherwise}.
    \end{cases}
\end{displaymath}

Note that sequence $a'[i]$ is strictly increasing and nonnegative. Moreover, for $i,j>0$,

\begin{align*}
    a'[i] + a'[j] & \le a'[i+j] & \iff \\
    C\cdot i + a[i] + C\cdot j + a[j] & \le  C(i+j) + a[i+j] & \iff \\
    a[i] + a[j] & \le  a[i+j]
    .
\end{align*}

When $i$ or $j$ equals $0$, then we have equality because $a'[0] = 0$.
\end{proof}

\begin{theorem}[\superadditivity $\rightarrow$ \uknapsack]
\label{superadditivity-uknapsack}
If \uknapsack can be solved in time $T(n,t)$,
then \superadditivity admits an algorithm with running time $\Oh\left(T(n,n)\log{n}\right)$.
\end{theorem}
\begin{proof}
Let $(a[i])_{i=0}^{n-1}$ be a nonnegative monotonic sequence (see Lemma~\ref{nonnegative-sequece}).
Set $D = \sum_{i=0}^{n-1} a[i] + 1$, and construct an \uknapsack instance with
the set of items $\left((i,a[i])\right)_{i=0}^{n-1}$ and
$((2n-1-i,D-a[i]))_{i=0}^{n-1}$ with target $t=2n-1$.
It is always possible to obtain $D$ by taking two items $(i,a[i])$, $(2n-1-i,D-a[i])$ for any $i$.
We claim that the answer to the constructed instance equals $D$ if and only if $a$ is superadditive.

If $a$ is not superadditive, then there are $i,j$ such that $a[i] + a[j] > a[i+j]$.
Choosing $(i,a[i])$, $(j,a[j])$, $(2n-1-i-j, D-a[i+j])$ gives a solution with a value
exceeding~$D$.

Now, assume that $a$ is superadditive.
Observe that any feasible knapsack solution may contain at most one item with a weight exceeding $n-1$.
On the other hand, the optimal solution has to include one such item
because the total value of the lighter ones is less than $D$.
Therefore, the optimal solution contains an item $(2n-1-k,D-a[k])$ for some $k < n$.
The total weight of the rest of the solution is at most $k$.
As $a$ is superadditive, we can replace any pair $(i,a[i]), (j,a[j])$
with the item $(i+j, a[i+j])$ without decreasing the value of the solution.
By repeating this argument, we end up with a single item lighter than $n$.
The sequence $a$ is monotonic; thus, it is always profitable to replace these two
items with the heavier one, as long as the load does not exceed $t$.
We conclude that every optimal solution must be of the form $((k,a[k]), (2n-1-k,D-a[k]))$,
which completes the proof.
\end{proof}

\begin{theorem}[\maxconvupper $\rightarrow$ \superadditivity]
\label{maxconvupper-superadditivity}
\label{thm:superadd-upperbound}
If \superadditivity can be solved in time $T(n)$,
then \textsc{MaxConv UpperBound} admits an algorithm with running time $\Oh\left(T(n)\log{n}\right)$.
\end{theorem}
\begin{proof}
We start by reducing the instance of \textsc{MaxConv UpperBound} to the
case of nonnegative monotonic sequences (analogous to
Lemma~\ref{nonnegative-sequece}). Observe that condition $a[i] + b[j] \le c[i+j]$ can be rewritten as
$(C + a[i] + Di) + (C + b[j] + Dj) \le 2C + c[i+j] + D(i+j)$ for any constants $C, D$.
Hence, replacing sequences $(a[i])_{i=0}^{n-1},\, (b[i])_{i=0}^{n-1},\, (c[i])_{i=0}^{n-1}$
with $a'[i] = C + a[i] + Di,\, b'[i] = C + b[i] + Di,\, c'[i] = 2C + c[i] + Di$ leads to an equivalent instance.
We can thus pick $C, D$ of magnitude $\Oh(W)$ to ensure that all elements are nonnegative
and that the resulting sequences are monotonic.
The values in the new sequences may increase to a maximum of $\Oh(nW)$.


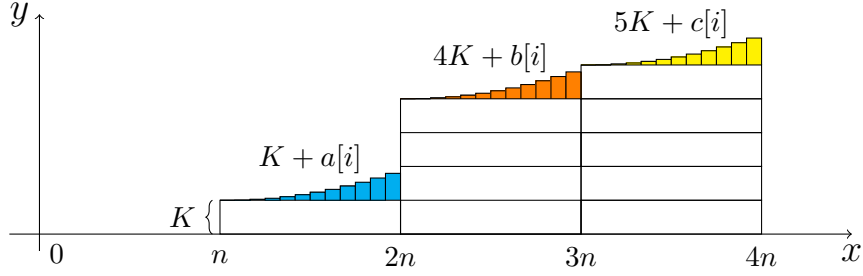
\begin{figure}[ht!]
    \centering
    \begin{tikzpicture}[xscale=0.8, yscale=0.45]

	   \draw [->] (-.5,0) -- ++(14,0) node[below] {\Large $x$};
	   \draw [->] (0,-.5) -- ++(0,7) node[left] {\Large $y$};
	   \foreach \i/\p/\a in {1/.5/2,3/.25/2.75} {%
		   \foreach \color/\transx/\transy in {cyan/3/1,orange/6/4, yellow/9/5} {
			   \begin{scope}
				   \foreach \x in {0,\p,...,\a} {%
					   \draw[fill=\color] (\x+\transx,\transy) -- (\transx+\x,{pow(\x,2.0)/9.5+\transy}) -- (\transx+\x+\p,{pow(\x,2.0)/9.5+\transy}) -- (\transx+\x+\p,\transy);
				   }
			   \end{scope}
			}
	   } 
       \draw[] (3,0) rectangle (6,1);
       \draw[] (6,0) rectangle (9,1);
       \draw[] (6,2) rectangle (9,3);
       \draw[] (6,0) rectangle (9,4);
       \draw[] (9,0) rectangle (12,1);
       \draw[] (9,1) rectangle (12,2);
       \draw[] (9,2) rectangle (12,3);
       \draw[] (9,3) rectangle (12,4);
       \draw[] (9,4) rectangle (12,5);

       \draw[below] (0,0) node[below right] () {$0$};
       \draw (3,0) node[yshift=-3mm] () {$n$};
       \draw (6,0) node[yshift=-3mm] () {$2n$};
       \draw (9,0) node[yshift=-3mm] () {$3n$};
       \draw (12,0) node[yshift=-3mm] () {$4n$};

       \draw[above] (4.5,1.5) node () {$K+a[i]$};
       \draw[above] (7.5,4.5) node () {$4K+b[i]$};
       \draw[above] (10.5,5.5) node () {$5K+c[i]$};

       \draw [decorate, decoration={brace, raise=1mm}] (3,0) --
       (3,1) node[midway,xshift=-5mm] {$K$};

   \end{tikzpicture}
    \label{rys:superadd-upperbound}
    \caption{Graphical interpretation of the sequence $e$ in Theorem \ref{thm:superadd-upperbound}.
    The height of rectangles equals~$K$.}
\end{figure}

Herein, we can assume the given sequences to be nonnegative and monotonic.
Define $K$ to be the maximum value occurring in given sequences $a,b,c$.
Construct a sequence $e$ of length $4n$ as follows.
For $i \in [0,n-1]$, set $e[i] = 0,\, e[n+i] = K + a[i],\, e[2n+i] = 4K + b[i],\, e[3n+i] = 5K + c[i]$.
If $a[i] + b[j] > c[i+j]$ exists for some $i,j$, then $e[n+i] + e[2n+j] > e[3n+i+j]$; therefore, $e$ is not superadditive.
We now show that in any other case, $e$ must be superadditive.

Assume w.l.o.g. that there are $i$ and $j$ such that $i \le j$.
The case $i < n$ can be ruled out because it implies $e[i] = 0$ and $e[i] + e[j] \le e[i+j]$ for any $j$, as $e$ is monotonic.
If $i \ge 2n$, then $i+j \ge 4n$; thus, we can restrict to $i \in [n, 2n-1]$.
For similar reasons, we can assume that $j < 3n$.
Now, if $j \in [n, 2n-1]$, then $e[i] + e[j] \le 4K \le e[i + j]$.
    Finally, for $j \in [2n, 3n-1]$, superadditivity clearly corresponds to \textsc{MaxConv UpperBound}'s defining condition.
\end{proof}

The proof of the reduction from \maxconv to \maxconvupper was recently independently
given in~\cite{sparsity}. The technique was introduced
by~\naiveciteone{Vassilevska and Williams}{vassilevska2010subcubic} to show
a subcubic reduction from $(\min,+)$-matrix multiplication for detecting a
negative weight triangle in a graph.

\begin{theorem}[\maxconv $\rightarrow$ \maxconvupper]
\label{maxconv-maxconvupper}
    A $T(n)$ algorithm for  \maxconvupper implies an $\Oh\left(T(\sqrt{n})n\log{n}\right)$ algorithm
    for \maxconv.
\end{theorem}

\begin{proof}
Let us assume that we have access to an oracle solving the \textsc{MaxConv UpperBound},
    i.e., checking whether $a \oplus^{\max} b \le c$.
First, we argue that by invoking this oracle $\log n$ times, we
can find an index $k$ for which there exists a pair $i,j$ violating the superadditivity constraint, i.e., satisfying $a[i]+b[j] > c[k]$, where $k = i+j$ if such an index $k$ exists.
Let $pre_k(s)$ be the $k$-element prefix of a sequence $s$.
The inequality $pre_k(a) \oplus^{\max} pre_k(b) \le pre_k(c)$ holds
only for those $k$ that are less than the smallest value of $i+j$ with a broken constraint.
We can use binary search to find the smallest $k$ for which the inequality does not hold.
This introduces an overhead of factor $\log {n}$.

Next, we want to show that by using an oracle that finds one violated index,
we can in fact find all violated indices.
Let us divide $[0,n-1]$ into $m = \sqrt{n} + \Oh(1)$ intervals $I_0, I_2, \dots I_m$ of equal length,
except potentially for the last one.
For each pair $I_x, I_y$, we can check whether $a[i] + b[j] \le c[i + j]$
for all $i \in I_x,\, j \in I_y$ and find a violated constraint (if any exist)
in time $T(\sqrt{n})\log{n}$
by translating the indices to $[0, 2n/m] = [0, 2\sqrt{n} + \Oh(1)]$.
After finding a pair $i,j$ that violates the superadditivity,
we substitute $c[i+j] := K$, where $K$ is a constant exceeding all feasible sums,
and continue analyzing the same pair.
Once anomalies are no longer detected, we move on to the next pair.
It is important to note that when an index $k$ violating superadditivity
is set to $c[k] := K$, this value $K$ is also preserved 
for further calls to the oracle -- in this way, we ensure that each violated index $k$
is reported only once.

For the sake of readability, we present a pseudocode (see Algorithm~\ref{alg:maxconv-detect}).
The subroutine \textsc{MaxConvDetectSingle} returns the value of $i+j$
for a broken constraint or $-1$ if none exist.
The notation $s^x$ stands for the subsequence of $s$ in the interval $I_x$.
    We assume that $c[i]=K$ for $i \ge n$.

\begin{algorithm}
	\caption{$\textsc{MaxConvDetectViolations}(a,b,c)$}
	\label{alg:maxconv-detect}
\begin{algorithmic}[1]
\For {$x=0,\ldots,m-1$}
  \For {$y=0,\ldots,m-1$}
    \State $k := 0$
    \While {$k \ge 0$}
      \State $k := \textsc{MaxConvDetectSingle}(a^x, b^y, c^{x+y} \cup c^{x+y+1})$
	  \If {$k \ge 0$}
	    \State $c[k] := K$
		\State violated$[k] :=$ true
	  \EndIf
	\EndWhile 
  \EndFor
\EndFor
\State \Return violated$[0,\dots,n-1]$
\end{algorithmic}
\end{algorithm}

The number of considered pairs of intervals equals $m^2 = \Oh(n)$.
Furthermore, for each pair, every call to \textsc{MaxConvDetectSingle} except the last one is followed
by setting a value of some element of $c$ to $K$.
This can happen only once for each element; hence, the total number of repetitions is at most $n$.
Therefore, the running time of the procedure \textsc{MaxConvDetectViolations} is $\Oh\left(T(\sqrt{n})n\log{n}\right)$.

By running this algorithm, we learn for each $k \in [0, n-1]$ whether
$c[k] > \max_{i \in [0, k]} a[i] + b[k-i]$.
Then, we can again use binary search
for each coordinate simultaneously.
After running the presented procedure $\log{W}$ times, the value of $c[k]$
will converge to $\max_{i \in [0, k]} a[i] + b[k-i]$ for every $k$.
\end{proof}

\begin{corollary}
If there exists a truly subquadratic algorithm for \maxconv, then it may be assumed to
have $\Ot(n)$ space dependency.
\end{corollary}
\begin{proof}

    Consider the Algorithm~\ref{alg:maxconv-detect}. It uses $\Oh(n)$ space to store
    the \textsc{violated} table containing the answer. The only other place where
    additional space might be required is the call to the \textsc{MaxConvDetectSingle} oracle. 
    Note that each call runs in time $T(\sqrt{n})$, as the parameters are tables 
    with $\Oh(\sqrt{n})$ elements. If \maxconv has a truly subquadratic algorithm, then $T(\sqrt{n}) =
    \Oh(n^{1-\eps/2})$, i.e.,\ it is truly sublinear. Because the oracle cannot
    use polynomially more space than its running time, the calls to the oracle
    require at most linear space (up to polylogarithmic factors).

This means that the main space cost of Algorithm~\ref{alg:maxconv-detect} is to
store an answer in the table \textsc{violated} and yields $\Ot(n)$ space
dependency.

\end{proof}

\section{The reduction from \knapsack to \maxconv}
\label{knapsack-bringmann}

We start with a simple observation: for \uknapsack (a single item can be chosen multiple times), an
$\Ot(t^2 + n)$ time algorithm can be obtained by using the standard dynamic
programming $\Oh(nt)$ algorithm.

\begin{theorem}
    There exists an $\Ot(t^2 + n)$ time algorithm for the \uknapsack
    problem.
\end{theorem}
\begin{proof}
    Our algorithm starts by discarding all items with weight larger than $t$.
    Since we are considering the unbounded case, for a given weight, we can ignore
    all items except the one with the highest value, as we can always take more
    copies of the most valuable item among the ones of equal weight. 
    We are left with at most $t$ items. Thus, using the standard
    $\Oh(nt)$ dynamic programming leads to a running time of $\Ot(t^2 + n)$.
\end{proof}

We show that from the perspective of the parameter $t$, this is the best we can hope for, unless $n$ appears in the complexity with an exponent higher than $2$
or there is a breakthrough for the \maxconv problem. In this section, we complement these results and show that a truly subquadratic algorithm for
\maxconv implies an $\Ot(t^{2-\epsilon} + n)$ algorithm for \knapsack.
We follow Bringmann's~\cite{bringmann} near-linear pseudopolynomial time
algorithm for \subsetsum and adapt it to the \knapsack problem. To do this, we
need to introduce some concepts related to the \subsetsum problem from previous
works. The key observation is that we can substitute the FFT in~\cite{bringmann} with \maxconv and consequently obtain
an $\Ot(T(t) + n)$ algorithm for \knapsack (where $T(n)$ is the time needed to solve \maxconv).

\subsection{Set of all subset sums}

Let us recall that in the \subsetsum problem, we are given a set $S$ of $n$ integers together with a target integer $t$. The goal is to determine whether there exists a subset of $S$ that sums up to $t$.

Horowitz~and~Sahni~\cite{sumset} introduced the notion of the set of \emph{all subset sums} that was later used
by Eppstein~\cite{sumset-app1} to solve the \emph{Dynamic Subset Sum} problem. More recently,
Koiliaris~and~Xu~\cite{faster-subsetsum} used it to develop an $\Ot(\sigma)$ algorithm for
\subsetsum ($\sigma$ denotes the sum of all elements). Later,
Bringmann~\cite{bringmann} improved this algorithm to $\Ot(n+t)$ ($t$ denotes the target
number in the \subsetsum problem).

The set of \emph{all subset sums} is defined as follows:

\begin{displaymath}
    \Sigma(S) = \Bigl\{ \sum_{a \in A} a \; |\; A \subseteq S \Bigr\}
    .
\end{displaymath}

For two sets $A, B \subseteq [0,u] $, the set $A \oplus B = \{ a+b \; |\; a \in A, b
\in B\}$ is their join, and $u$ is the upper bound of the elements $A$ and $B$. This join can be computed in time $\Oh(u\log u)$ by using the FFT. Namely, we write $A$ and $B$ as polynomials $f_A(x) =
\sum_{i\in A} x^i$ and $f_B(x) = \sum_{i\in B} x^i$, respectively. Then, we can compute the polynomial $g =
f_1 \cdot f_2$ in $\Oh(u\log u)$ time. Polynomial $g$ has a nonzero coefficient in front of
the term $x^i$ iff $i \in A \oplus B$. We can also easily extract $A\oplus B$.

Koiliaris~and~Xu~\cite{faster-subsetsum} noticed that if we want to compute $\Sigma(S)$ for a
given $S$, we can partition $S$ into two sets: $S_1$ and $S_2$, recursively
compute $\Sigma(S_1)$ and $\Sigma(S_2)$, and then join them using the FFT.
Koiliaris~and~Xu~\cite{faster-subsetsum} analyzed their algorithm using
Lemma~\ref{log-analysis}, which was later also used by Bringmann~\cite{bringmann}.

\begin{lemma}[\cite{faster-subsetsum}, Observation 2.6]
    \label{log-analysis}
        Let $g$ be a positive, superadditive (i.e., $\forall_{x,y} g(x + y) \ge g(x) + g(y)$)
    function. For a function $f(n,m)$ satisfying
    \begin{displaymath}
        f(n,m) = \max_{m_1 + m_2 = m} \Bigl\{ f\Bigl(\frac{n}{2}, m_1\Bigr) +
    f\Bigl(\frac{n}{2}, m_2\Bigr) + g(m)) \Bigr\}
    \end{displaymath}

    we have that $f(n,m) = \Oh(g(m)\log n)$.
\end{lemma}


\subsection{Sum of all sets for \knapsack}
\label{piknapsack}

We now adapt the notion of the sum of all sets to the \knapsack setting. 
Here, we use a data structure that, for a given capacity, stores the value of the best
solution we can pack. This data structure can be implemented as an array of
size $t$ that keeps the largest value in each cell (for comparison, $\Sigma(S)$
was implemented as a binary vector of size $t$). To emphasize that we are
working with \knapsack, we use $\Pi(S)$ to denote the array of the values
for the set of items $S$.


If we have two partial solutions $\Pi(A)$ and $\Pi(B)$, we can compute their join,
denoted as $\Pi(A) \oplus^{\max} \Pi(B)$. A valid solution in $\Pi(A)
\oplus^{\max} \Pi(B)$ of weight $t$ consists of a
solution from $\Pi(A)$ and one from $\Pi(B)$ that sum up to $t$ (one of them can
    be $0$). Hence, $\Pi(A) \oplus^{\max} \Pi(B)[k] = \max_{0 \le i \le k} \{ \Pi(A)[k-i] +
\Pi(B)[i] \}$. This product is the \maxconv of array $\Pi(A)$ 
and $\Pi(B)$. We will use $\Pi(A) \oplus_t^{\max} \Pi(B)$ to denote the
\maxconv of $A$ and $B$ for domain~$\{0,\ldots,t\}$.

To compute $\Pi(S)$, we can split $S$ into two equal-cardinality, disjoint
subsets $S = S_1 \cup S_2$, recursively compute $\Pi(S_1)$ and $\Pi(S_2)$, and
finally join them in $\Oh(T(\sigma))$ time ($\sigma$ is the sum of weights of
all items). By Lemma~\ref{log-analysis}, we obtain an $\Oh(T(\sigma)\log{\sigma}
\log{n})$ time algorithm (recall that the naive algorithm for \maxconv works in
$\Oh(n^2)$ time).

\subsection{Retracing Bringmann's steps}

In this section, we obtain an $\Ot(T(t)+n)$ algorithm for \knapsack, which improves
upon the $\Ot(T(\sigma))$ algorithm from the previous section.
In his algorithm~\cite{bringmann} for \subsetsum, Bringmann uses two key techniques.
First, \emph{layer splitting} is based on a very useful observation that an instance $(Z,t)$ can
be partitioned into $\Oh(\log{n})$ layers $L_i \subseteq (t/2^i, t/2^{i-1}]$
(for $0 < i < \ceil{\log{n}}$) and
$L_{\ceil{\log{n}}} \subseteq [0,t/2^{\ceil{\log{n}}-1}]$. With this partition, we may infer that for $i>0$,
at most $2^i$ elements from the set $L_i$ can be used in any solution 
(otherwise, their cumulative sum would be larger than $t$).
The second technique is an application of \emph{color coding}~\cite{color-coding} 
that results in a fast, randomized algorithm that can compute all solutions with a sum of at
most $t$ using no more than $k$ elements. 
By combining those two techniques, Bringmann~\cite{bringmann} developed an $\Ot(t + n)$
time algorithm for \subsetsum. We now retrace both ideas and use them in
the \knapsack context.

\subsubsection{Color Coding}

We modify Bringmann's~\cite{bringmann} color coding technique by using \maxconv
instead of FFT to obtain an algorithm for \knapsack.
We first discuss the Algorithm~\ref{alg:colorcoding}, which can compute all
solutions in $[0,t]$ that use at most $k$ elements with high probability. We start
by randomly partitioning the set of items into $k^2$ disjoint sets $Z = Z_1 \cup \ldots \cup Z_{k^2}$.
Algorithm~\ref{alg:colorcoding} succeeds in finding a given solution if
its elements are placed in different sets of the partition $Z$.


\begin{lemma}
    \label{color-coding-lemma}

    There exists an algorithm that computes an array $W$ in time $\Oh(T(t) k^2
    \log{(1/\delta)})$ such that, for any $Y \subseteq Z$ with $|Y| \le k$ and
    every weight $i \in [0,t]$, we have $\Pi(Y)[i] \le W[i] \le \Pi(Z)[i]$ with probability $\ge
    1-\delta$ for any constant $\delta \in (0,1)$ (where $T(n)$ is the time needed to compute \maxconv).

\end{lemma}

\begin{algorithm}
    \caption{$\textsc{ColorCoding}(Z,t,k,\delta)$ (cf.~\cite[Algorithm
    1]{bringmann}).}
    \label{alg:colorcoding}
\begin{algorithmic}[1]
    \For {$j=1,\ldots,\ceil{\log_{4/3}(1/\delta)}$}
  \State randomly partition $Z = Z_1 \cup \ldots \cup Z_{k^2}$
  \State $\textbf{P}_j = Z_1 \oplus_t^{\max} \ldots \oplus_t^{\max} Z_{k^2}$
\EndFor
\State \Return $W$, where $W[i] = \max_{j} \textbf{P}_j[i]$
\end{algorithmic}
\end{algorithm}

\begin{proof}

We show split $Z$ into $k^2$ parts: $Z_1 \cup \ldots \cup Z_{k^2}$.
Here, $Z_i$ is an array of size $t$, and $Z_i[j]$ is the value of a single element
(if one exists) with weight $j$ in $Z_i$ (in case of a conflict, we select a random one).

We claim that $Z_1 \oplus_t^{\max} \ldots \oplus_t^{\max} Z_{k^2}$ contains solutions at least as good as those that
use $k$ items (with high probability).
We use the same argument as in~\cite{bringmann}. Assume that
the best solution uses the set $Y \subseteq Z$ of items and $|Y| \le k$. The probability that
all items of $Y$ are in different sets of the partition is the same as the probability
that the second element of $Y$ is in a different set than the first one, the
third element is in a different set than the first and second item, etc. That is:

\begin{displaymath}
    \frac{k^2 - 1}{k^2} \cdot \frac{k^2 - 2}{k^2} \ldots \frac{k^2 - (|Y| -
            1)}{k^2} \ge \bigg(1 - \frac{(|Y| - 1)}{k^2}\bigg)^{|Y|} \ge \bigg(1 -
        \frac{1}{k}\bigg)^k \ge \bigg(\frac{1}{2}\bigg)^2  = \frac{1}{4}.
\end{displaymath}

By repeating this process $\Oh(\log(\frac{1}{\delta}))$ times, we obtain the
correct solution with a probability of at least $1-\delta$. Also, to compute
\maxconv, we need $k^2$ repetitions. Hence, we obtain an $\Oh(T(t) k^2 \log(1/\delta))$ time
algorithm.  \end{proof}


\subsubsection{Layer Splitting}

We can split our items into $\log{n}$ layers. Layer $L_i$ is the set of items with weights in $(t/2^i, t/2^{i-1}]$ for $0 < i < \ceil{\log{n}}$; the last layer $L_{\ceil{\log{n}}}$ has items with weights in $[0,t/2^{\ceil{\log{n}}-1}]$. With this, we can be sure that only $2^i$ items from
the layer $i$ can be chosen for a solution. If we can quickly compute $\Pi(L_i)$ for all $i$, 
then it suffices to compute their \maxconv $\Oh(\log{n})$ times.
We now show how to compute $\Pi(L_i)$ in $\Ot(T(t) + n)$ time using
color coding.

\begin{lemma}
    For all $i$, there exists an algorithm that, for $L_i \subseteq (\frac{t}{2^i},
    \frac{t}{2^{i-1}}]$
    and for all $\delta \in (0,1/4]$, computes $\Pi(L_i)$ in
    $\Oh(T(t\log{t}\log^3(2^{i-1}/\delta)))$ time, where each entry of $\Pi(L_i)$ is correct
    with a probability of at least $1-\delta$.
\end{lemma}

\begin{algorithm}
    \caption{$\textsc{ColorCodingLayer}(L,t,i,\delta)$ (cf. \cite[Algorithm
    3]{bringmann}).}\label{alg:colorcodinglayer}
\begin{algorithmic}[1]
\State $l = 2^i$
\State \textbf{if} $l < \log(l/\delta)$ \textbf{then return} $\textsc{ColorCoding}(L,t,l,\delta)$
\State $m = l/\log(l/\delta)$ rounded up to the next power of 2
\State randomly partition $L = A_1 \cup \ldots \cup A_{m}$
\State $\gamma = 6 \log(l/\delta)$
\For {$j=1,\ldots,m$}
 \State $\textbf{P}_j = \textsc{ColorCoding}(A_j,2 \gamma t/l, \gamma, \delta/l)$
\EndFor
\For {$h=1,\ldots,\log m$}
  \For {$j=1,\ldots,m / 2^{h}$}
  \State $\textbf{P}_j = \textbf{P}_{2j-1} \oplus_{2^h \cdot 2 \gamma t/l}^{\max} \textbf{P}_{2j}$
  \EndFor
\EndFor
\State \Return $\textbf{P}_1$
\end{algorithmic}
\end{algorithm}

\begin{proof}

We use the same arguments as in~\cite[Lemma 3.2]{bringmann}. First, we
split the set $L$ into $m$ disjoint subsets $L = A_1 \cup \ldots \cup A_m$ 
    (where $m = l/\log(l/\delta)$). Then, for every partition, we compute $\Pi(A_i)$ using
$\Oh(\log(l/\delta))$ items and probability $\delta/l$
using~Lemma~\ref{color-coding-lemma}. For every $A_i$,
$\Oh(T(\log(l)t/l)\log^3(l/\delta))$ time is required. Hence, for all $A_i$, we need
$\Oh(T(t)\log^3(l/\delta))$ time, as \minconv needs at
least linear time $T(n) = \Omega(n)$.

Ultimately, we need to combine arrays $\Pi(A_i)$ in a ``binary tree way''.
In the first round, we compute $\Pi(A_1) \oplus^{\max} \Pi(A_2), \Pi(A_3)
\oplus^{\max} \Pi(A_4),\ldots, \Pi(A_{m-1}) \oplus^{\max} \Pi(A_m)$. Then, in the second round, we join the products of
the first round in a similar way. We continue until we have joined all subsets.
This process yields us significant savings over just computing $\Pi(A_1) \oplus^{\max}
    \ldots \oplus^{\max} \Pi(A_m)$ because in round $h$, we need to compute
    \maxconv with numbers of order $\Oh(2^{h} t \log(l/\delta)/l)$, and there are at most
$\log{m}$ rounds. The complexity of joining them is as follows:

\begin{displaymath}
\sum^{\log{m}}_{h=1} \frac{m}{2^h} T(2^h \log(l/\delta) t/l) \log{t}) =
\Oh(T(t\log{t})\log{m}).
\end{displaymath}

Overall, we determine that the time complexity of
    the algorithm is $\Oh(T(t\log{t})\log^3(l/\delta))$
    (some logarithmic factors could be omitted if we assume that there exists
$\epsilon > 0$ such that $T(n) = \Omega(n^{1+\epsilon})$).

The correctness of the algorithm is based on~\cite[Claim 3.3]{bringmann}.
We take a subset of items $Y \subseteq L$ and let $Y_j = Y\cap A_j$. Claim 3.3 in~\cite{bringmann}
says that $\prob[|Y_j| \ge 6\log(l/\delta)] \le \delta/l$. Thus, we
can run ColorCoding procedure for $k=6\log(l/\delta)$ and still guarantee
a sufficiently high probability of success.

\end{proof}

\begin{theorem}[\knapsack $\rightarrow$ \maxconv]
    If \maxconv can be solved in $T(n)$ time, then 
    \knapsack can be solved in time $\Oh(T(t\log{t}) \log^3(n/\delta)\log{n})$
    with a probability of at least $1-\delta$.
\label{knapsack-maxconv}
\end{theorem}

\begin{algorithm}
    \caption{$\mathrm{Knapsack}(Z,t,\delta)$ (cf. \cite[Algorithm
    2]{bringmann}).}
\label{alg:knapsack}
\begin{algorithmic}[1] 
\State split $Z$ into $L_i = Z \cap (t/2^{i},t/2^{i-1}]$ for
$i=1,\ldots,\ceil{\log n} - 1$, and $L_{\ceil{\log n}} = Z \cap [0,t/2^{\ceil{\log n } -1}]$ 
\State $\textbf{W} = \emptyset$
\For {$i=1,\ldots, \ceil{\log n}$}
\State $\textbf{P}_i = \textsc{ColorCodingLayer}(L_i,t,i,\delta / \ceil{\log n})$
\State $\textbf{W} = \textbf{W} \oplus^{\max} \textbf{P}_i$
\EndFor
\State \Return $\textbf{W}$
\end{algorithmic}
\end{algorithm}

\begin{proof}

To obtain an algorithm for \knapsack, as mentioned before, we need to split $Z$ into
disjoint layers $L_i = Z \cap (t/2^i, t/2^{i-1}]$ and $L_{\ceil{\log{n}}} =
Z \cap [0,t/2^{\ceil{\log{n}} - 1}]$. Then, we compute $\Pi(L_i)$ for all $i$ and
join them using \maxconv. We present the pseudocode in
Algorithm~\ref{alg:knapsack}. It is based on~\cite[Algorithm 2]{bringmann}. Overall, $\Oh(T(t\log{t})
\log^3{(n/\delta)} \log{n} + T(t)\log{n}) = \Oh(T(t\log{t}) \log^3(n/\delta)\log{n})$ time is required. 
\end{proof}

Koiliaris~and~Xu~\cite{faster-subsetsum} considered a variant of \subsetsum where one needs to
check if there exists a subset that sums up to $k$ for all $k \in [0,t]$. Here,
we note that a similar extension for \knapsack is also equivalent to \maxconv
(see Section~\ref{problems} for the definition of \knapsackplus).

\begin{corollary}[\knapsackplus $\rightarrow$ \maxconv]
    If \maxconv can be solved in $T(n)$ time, then 
    \knapsackplus can be solved in $\Oh(T(t\log{t}) \log^3(t n/\delta)\log{n})$
    time with a probability of at least $1-\delta$.
\end{corollary}

Algorithm~\ref{alg:knapsack} returns an array
$\Pi(Z)$, where each entry $z \in \Pi(Z)$ is optimal with probability
$1-\delta$. Now, if we want to obtain the optimal solution for all knapsack
capacities in $[1,t]$, we need to increase the success probability to
$1-\frac{\delta}{t}$ so that we can use the union bound. 
Consequently, in this case, a single entry is faulty with a probability of at most $\delta/t$,
and we can upper bound the event, where at least one entry is incorrect
by $\frac{\delta}{t} t = \delta$. This introduces an additional $\mathrm{polylog}(t)$ factor
in the running time.

Finally, for completeness, we note that \knapsackplus is more general than
\knapsack. \knapsackplus returns a solution for all capacities $\le t$.
However, in the \knapsack problem, we are interested only in a capacity equal to exactly
$t$. 

\begin{corollary}[\knapsack $\rightarrow$ \knapsackplus]
    If \knapsackplus can be solved in $T(t,n)$ time, then \knapsack can be solved
    in $\Ot(T(t,n))$ time.
\end{corollary}

The next corollary follows from the ring of reductions.

\begin{corollary}
An $\Ot((n+t)^{2-\eps})$ time algorithm for \knapsack implies
an $\Ot(t^{2-\eps'} + n)$ time algorithm for \knapsackplus.
\end{corollary}

\section{Other problems related to \minconv}
\label{sec:other}

\subsection{Maximum Consecutive Subsums Problem}
\label{sec:mcsp}

The \textsc{Maximum Consecutive Subsums Problem} (\mcsp) is, to the best of our
knowledge, the first problem explicitly proven to be nontrivially subquadratically equivalent to \textsc{MinConv}~\cite{mscp}.
In this section, we show the reduction from \mcsp to \maxconv for
completeness. Moreover, we present the reduction in the opposite direction,
which, in our opinion, is simpler than the original one.

\begin{theorem}[\mcsp $\rightarrow$ \maxconv]
If \maxconv can be solved in time $T(n)$,
then \mcsp admits an algorithm with running time $\Oh\left(T(n)\right)$.
\end{theorem}
\begin{proof}
Let $(a[i])_{i=0}^{n-1}$ be the input sequence.
Construct sequences of length $2n$ as follows: $b[k] = \sum_{i=0}^{k} a[i]$ for $k < n$ and $c[k] = -\sum_{i=0}^{n-k-1} a[i]$ for $k \le n$ (empty sum equals 0); otherwise, $b[k] = c[k] = -D$, where $D$ is two times larger than any partial sum.
Observe that
\begin{equation}
(b \oplus^{\max} c)[n+k-1] = \max_{0 \le j < n \atop 0 \le n+k-j-1 \le n} \sum_{i=0}^{j} a[i] - \sum_{i=0}^{j-k} a[i] = \max_{k - 1 \le j < n} \sum_{i=j-k+1}^{j} a[i].
\end{equation}
Thus, we can determine the maximum consecutive sum for each length $k$ after performing \maxconv.
\end{proof}

\begin{theorem}[\superadditivity $\rightarrow$ \mcsp]
If \mcsp can be solved in time $T(n)$,
then \superadditivity admits an algorithm with running time $\Oh\left(T(n)\right)$.
\end{theorem}
\begin{proof}
Let $(a[i])_{i=0}^{n-1}$ be the input sequence and $b[i] = a[i+1] - a[i]$.
The superadditivity condition $a[k] \le a[k+j] - a[j]$ (for all possible $k,j$)
can be translated into $a[k] \le \min_{0 \le j < n-k} \sum_{i=j}^{k+j-1} b[i]$ (for all $k$). Thus, computing the \mcsp vector on $(-b[i])_{i=0}^{n-2}$ is sufficient to verify whether the above condition holds.
\end{proof}

\subsection{Tree sparsity}
\label{sec:sparsity}

\begin{theorem}[\sparsity $\rightarrow$ \maxconv]
\label{thm:sparsity}
If \maxconv can be solved in time $T(n)$ and the function $T$
is superadditive,
then \sparsity admits an algorithm with running time $\Oh\left(T(n)\log^2{n}\right)$.
\end{theorem}
\begin{proof}
We take advantage of the heavy-light decomposition introduced
by~\naiveciteone{Sleator and Tarjan}{heavy-light}.
This technique has been utilized by~\naivecitemany{Backurs}{sparsity} to transform a nearly linear PTAS
for \maxconv to a nearly linear PTAS for \sparsity. 

We decompose a tree into a set of paths (which we call \emph{spines}) that will start from a \emph{head}.
First, we construct a \emph{spine} with a \emph{head} $s_1$ at the root of the tree.
We define $s_{i+1}$ as the child of $s_i$ for a larger subtree (in case of a draw, we choose any child)
and the last node in the spine as a leaf.
The remaining children of node $s_i$ become heads for analogous spines such that the whole tree is covered.
Note that every path from a leaf to the root intersects at most $\log{n}$ spines because
each spine transition doubles the subtree size.

\tikzstyle{vertex}=[circle,fill=black!25,minimum size=5pt]
\tikzstyle{selected vertex} = [vertex, fill=blue!25]
\tikzstyle{edge} = [draw,thick,-]
\tikzstyle{selected edge} = [draw,line width=5pt,blue!50,cap=round]

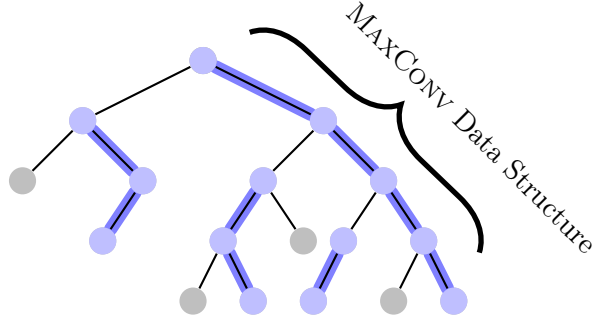
\begin{figure}[ht!]
\centering
\begin{tikzpicture}[scale=0.8]

    \begin{scope}
	\node[vertex] (a) at (0,0) {};
    \foreach \father/\direction/\level/\name in {{a/1/1/b},{a/-1/1/c}, 
                        {b/1/2/d},{b/-1/2/e},{c/1/2/f},{c/-1/2/g}, 
                        {d/1/3/dr},{d/-1/3/dl},{e/1/3/er},{e/-1/3/el}, 
                        {dr/1/4/drr},{dr/-1/4/drl},{dl/-1/4/dll}, 
                        {el/1/4/elr},{el/-1/4/ell}, 
                        {f/-1/3/fl}}
        {
        \node[vertex] (\name) at
        ([shift=({2*\direction/\level,-1})]\father) {};
        }

    \foreach \a/\b in {{a/b},{b/d},{d/dr},{dr/drr},{dl/dll},{e/el},{el/elr},{c/f},{f/fl}}
        {
            \path[selected edge] (\a) -- (\b);
        }

	\node[vertex] (a) at (0,0) {};
    \foreach \father/\direction/\level/\name in {{a/1/1/b},{a/-1/1/c}, 
                        {b/1/2/d},{b/-1/2/e},{c/1/2/f},{c/-1/2/g}, 
                        {d/1/3/dr},{d/-1/3/dl},{e/1/3/er},{e/-1/3/el}, 
                        {dr/1/4/drr},{dr/-1/4/drl},{dl/-1/4/dll}, 
                        {el/1/4/elr},{el/-1/4/ell}, 
                        {f/-1/3/fl}}
        {
        \path[edge] (\father) edge  (\name);
        }
    \foreach \a in {a,b,d,dr,drr,e,el,elr,dl,dll,c,f,fl}
        \node[selected vertex] at (\a) {};

        \draw [decorate, line width=2pt, decoration={brace, raise=7mm,
        amplitude=7mm}, above=1cm] (a) -- (drr) node [midway,
    above=5mm,xshift=17mm, rotate=-45] {\maxconv Data Structure} ;
    \end{scope}
	
\end{tikzpicture}
\caption{Schema of spine decomposition~\cite{sparsity}. Blue edges represent
    edges on the spine. For each spine, we build
    an efficient data structure that uses \maxconv (curly brackets). There are at most $\Oh(\log{n})$
different spines on any path from a leaf to the root.}
\end{figure}

Similar to~\cite{sparsity} for a node $v$ with a subtree of size $m$, we
want to compute the \emph{sparsity vector} $U = (U[0], U[1],\ldots,U[m])$,
where the index $U[i]$ represents the weight of the heaviest subtree rooted at $v$ with size $i$.
We compute sparsity vectors for all heads of spines in the tree recursively.
Let $(s_i)_{i=1}^\ell$ be a spine with a head $v$, and for all $i$, let $U^i$ indicate
the sparsity vector for the child of $s_i$ that is a head (i.e., the child with the smaller subtree).
If $s_i$ has less than two children, then $U^i$ is a zero vector.

For an interval $[a, b] \subseteq [1, \ell]$, let
$U^{a, b} = U^a \oplus^{\max} U^{a+1} \oplus^{\max} \dots \oplus^{\max} U^b$, and
    let $Y^{a, b}[k]$ be a vector such that for all $k$, $Y^{a,b}[k]$ is the
    weight of a subtree of size $k$ rooted at $s_{a}$ and not containing $s_{b+1}$ (if it exists). 
Let $c = \floor{\frac{a+b}{2}}$.
The $\oplus^{\max}$ operator is associative; hence, $U^{a, b} = U^{a, c} \oplus^{\max} U^{c+1, b}$.
To compute the vector $Y^{a,b}$, we consider two cases, depending on whether the optimal subtree contains $s_{c+1}$.

\begin{eqnarray*}
Y^{a, b}[k] &=& \max\bigg[Y^{a, c}[k],\quad \sum_{i=a}^c w(s_i) + \max_{k_1 + k_2 = k - (c - a + 1)} \Big(U^{a, c}[k_1] + Y^{c+1, b}[k_2] \Big)\bigg] \\
&=& \max\bigg[Y^{a, c}[k],\quad \sum_{i=a}^c w(s_i) + \Big(U^{a, c} \oplus^{\max} Y^{c+1, b}\Big)\big[k - (c - a + 1)\big] \bigg]
\end{eqnarray*}

Recall, that $w : V(T) \rightarrow \mathbb{N}_{\ge 0}$ is the weight
function from the definition of the problem (see Section~\ref{sec:otherproblems}). Using the presented formulas, we reduce the problem of computing $X^v = Y^{1, \ell}$ to
subproblems for intervals $[1, \frac{\ell}{2}]$ and $[\frac{\ell}{2}+1, \ell]$, and we merge the results with two $(\max,+)$-convolutions.
Proceeding further, we obtain $\log{\ell}$ levels of recursion, where the sum of convolution sizes
on each level is $\Oh(m)$, which results in a total running time of $\Oh\left(T(m)\log{m}\right)$
(recall that $T$ is superadditive).


The heavy-light decomposition guarantees that there are at most $\Oh(\log{n})$
different spines on a path from a leaf to the root. Moreover, we compute
sparsity vectors for all heads of the spine, with at most $\log{n}$ levels of
recursion. In each recursion, we execute the \maxconv procedure. Hence,
we obtain a running time of $\Oh(T(n)\log^2{n})$.
\end{proof}

\subsection{$l_\infty$-Necklace Alignment}
\label{sec:necklace}

In this section, we study the \necklace alignment problem,
which has been shown to be reducible to \minconv~\cite{necklaces}.
Even though we were not able to prove it as equivalent to \minconv,
we have observed that \necklace is tightly connected to the $(\min,+)$-convolution, which leads to a reduction from a related problem -- \maxconvlower.
This opens an avenue for expanding the class of problems equivalent to \minconv; however, it turns out that we first need to better understand the nondeterministic complexity of \minconv.
We elaborate on these issues in this and the following section.

\begin{theorem}[\maxconvlower $\rightarrow$ \necklace]
\label{thm:necklace}
If \necklace can be solved in time $T(n)$,
then \maxconvlower admits an algorithm with running time $\Oh\left(T(n)\log{n}\right)$.
\end{theorem}
\begin{proof}
Let $a,b,c$ be the input sequences for \maxconvlower.
    A \emph{combination} is the sum of any choice of $m$ elements from these sequences. More
    formally:

    \begin{definition}[combination]
        A \emph{combination} of length $m$ is a sum:

        \begin{displaymath}
            f_1 \cdot e_1[k_1] + f_2 \cdot e_2[k_2] + \ldots f_m \cdot e_m[k_m]
            ,
        \end{displaymath}

        where $e_i \in \{a,b,c\}$, $f_i \in \{-1, 1\}$ and $k_i \in \{0,\ldots, n-1\}$.

        The \emph{order} of this combination is as follows:

        \begin{displaymath}
            \sum_{i=1}^m f_i \cdot k_i
            .
        \end{displaymath}
    \end{definition}


We can assume the following properties of the input sequences w.l.o.g.
\begin{enumerate}
\item \emph{We may assume that the sequences are nonnegative and that $a[i] \le c[i]$ for all $i$.}
To guarantee this, we add $C_1$ to $a$, $C_1+C_2$ to $b$, and $2C_1+C_2$ to $c$ for appropriate positive constants $C_1,C_2$.

\item \emph{We can assume that the combinations of order $\le n$ that contain the last element
    of sequence $b$ with a positive coefficient are positive.} We can achieve this property by artificially
    appending any $b[n]$ that is larger than the sum of all elements. Note that
    since it is the last element, it does not influence the result of the \maxconvlower instance.

\item \emph{Any combination of positive order and length bounded by $L$
has a nonnegative value.} One can guarantee this by adding a linear function $Di$ to all sequences.
As the order of the combination is positive, the factors at $D$ sum up
to a positive value.
It suffices to choose $D$ equal to the maximum absolute value of an element
        times a parameter $L$ that will be set to 10.
Note that previous inequalities compare combinations of the same order, 
and so they remain unaffected.
\end{enumerate}

These transformations might increase the values of the elements to $\Oh(nWL^2)$.
Let $B = b[n],\, B_1 = b[n-1],\, B_2 = b[n] - b[1]$.
We define necklaces $x,y$ of length $2B$ with $N=2n$ beads each.

\begin{equation*}
\resizebox{\hsize}{!}{
$\begin{array}{ccccccccccc}
x = \big(&a[0],& a[1],& \dots,& a[n-1],& B + c[0],& B + c[1],& \dots,& B + c[n-2],& B + c[n-1] &\big), \\
y = \big(&B_1 - b[n-1],& B_1 - b[n-2],& \dots,& B_1 - b[0],& B+B_2-b[n-1],& B+B_2-b[n-2],& \dots,& B+B_2-b[1],& 2B&\big).
\end{array}$}
\end{equation*}

Property (3) implies monotonicity of the sequences because for any
$0\le i<j \le n$, the combination $a[j] - a[i]$ is greater than zero. 

Let $d(x[i],y[j])$ be the forward distance between $x[i]$ and $y[j]$, i.e.,
$y[j] - x[i]$ plus the length of the necklaces if $j < i$.
For all $k$, define $M_k$ to be $\max_{i \in [0,N)} \dxy - \min_{i \in [0,N)} \dxy$.
In this setting, \cite[Fact 5]{necklaces} says
    that for a fixed $k$, the optimal solution has a value of $\frac{M_k}{2}$.

We want to show that for $k \in [0,n)$, the following holds:
\begin{eqnarray*}
\min_{i \in [0,2n)} \dxy &=& B_1 - \max_{i + j\, =\, n-k-1} (a[i] + b[j]), \\
\max_{i \in [0,2n)} \dxy &=& B - c[n-k-1].
\end{eqnarray*}


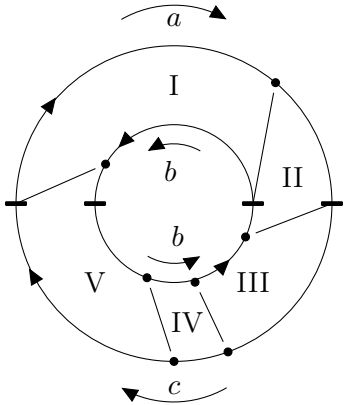
\begin{figure}[ht!]
    \centering
    \begin{tikzpicture}[scale=0.7,
                        >=triangle 45,
                       mydecor/.style = {decoration = {markings, 
                                                       mark = at position #1 with {\arrow{<}}}
                                       },
                       mydecol/.style = {decoration = {markings, 
                                                       mark = at position #1 with {\arrow{>}}}
                                       }
                      ]
     \draw[postaction = {mydecor=0.4 ,decorate}, 
           postaction = {mydecor=0.58333 ,decorate}] (0,0) circle (3cm);

     \draw[postaction = {mydecol=0.375 ,decorate}, 
           postaction = {mydecol=0.875 ,decorate}] (0,0) circle (1.5cm);

     \draw[->] (0.5,1) arc (60:120:1cm) node[below,xshift=0.3cm] {$b$};
     \draw[->] (-0.5,-1) arc (-120:-60:1cm) node[above,xshift=-0.3cm] {$b$};
     \draw[->] (-1,3.5) arc (120:60:2cm) node[xshift=-0.7cm] {$a$};
     \draw[->] (1,-3.5) arc (-60:-120:2cm) node[xshift=0.7cm] {$c$};

   \draw[line width=1.5pt] (-3,0) ellipse (.2cm and 0cm);
   \draw[line width=1.5pt] (3,0) ellipse (.2cm and 0cm);
   \draw[line width=1.5pt] (-1.5,0) ellipse (.2cm and 0cm);
   \draw[line width=1.5pt] (1.5,0) ellipse (.2cm and 0cm);

	\draw[fill = black] (-1.2994, 0.75) node (wew1) {} circle (0.075cm);
    \draw[fill = black] (1.35946,-0.633927) node (wew2) {} circle (.075cm);
    \draw[fill = black] (0.4, -1.4942) node (wew3) {} circle (0.075cm);
	\draw[fill = black] (-0.51303, -1.40954) node (wew4) {} circle (0.075cm);

	\draw[fill = black] (1.928362, 2.2981333) node (zew1) {} circle (0.075cm);
	\draw[fill = black] (1.02606, -2.819077) node (zew2) {} circle (0.075cm);
	\draw[fill = black] (0, -3) node (zew3) {} circle (0.075cm);
	
	\draw (wew1) -- (-3,0);
	\draw (1.5,0) -- (zew1);
	\draw (wew2) -- (3,0);
	\draw (wew3) -- (zew2);
	\draw (wew4) -- (zew3);
	
    \draw[fill=white] (0,2.25) circle (0.0cm) node {I};
    \draw[fill=white] (2.25,0.5) circle (0.0cm) node {II};
    \draw[fill=white] (1.5,-1.5) circle (0.0cm) node {III};
    \draw[fill=white] (0.25,-2.25) circle (0.0cm) node {IV};
    \draw[fill=white] (-1.5,-1.5) circle (0.0cm) node {V};
   %

   \end{tikzpicture}
   \caption{Five areas that correspond to the five types of connections between
     beads. The inner circle represents
   two repetitions of the sequence $b$. The outer circle consists of the sequence $a$ and then
     the sequence $c$.}
   \label{fig:necklace}
\end{figure}


There are five types of connections between beads (see Figure~\ref{fig:necklace}).
\begin{align*}
\dxy = \left\{\begin{array}{lll}
B_1 - a[i] - b[n-k-1-i] & \quad i \in [0,n-k-1], & \mathrm{(I)} \\
B+B_2-a[i]-b[2n-k-1-i] & \quad i \in [n-k,n-1], & \mathrm{(II)}\\
B_2 - b[2n-k-1-i] - c[i-n] & \quad i \in [n,2n-k-2], & \mathrm{(III)}\\
B - c[n-k-1]& \quad i = 2n-k-1, &\mathrm{(IV)} \\
B + B_1  - b[3n-k-1-i] - c[i-n] & \quad i \in [2n-k,2n-1]. & \mathrm{(V)}
\end{array}
\right.
\end{align*}

All formulas form combinations of length bounded by 5; thus, we can apply properties (2) and (3).
Observe that the order of each combination equals $k$, except for $i = 2n-k-1$,
where the order is $k+1$.
Using property (3), we reason that $B - c[n-k-1]$ is indeed the maximal forward distance.
We now show that the minimum lies within the group~$\mathrm{(I)}$.
First, note that these are the only combinations with no occurrences of $b[n]$.
We claim that every distance in group~$\mathrm{(I)}$ is upperbounded by all distances in other groups.
This is clear for group $\mathrm{(IV)}$ because the orders differ. For other groups, we can use property (2), as the
combinations in question have the same order and only the one not in group~$\mathrm{(I)}$ contains $b[n]$.

For $k < n$, the condition $M_k < B - B_1$ is equivalent
to $c[n-k-1] > \max_{i + j\, =\, n-k-1} (a[i] + b[j])$.
If such a $k$ exists, i.e., the answer to \maxconvlower for sequences $a,b,c$ is NO,
then $\min_k M_k < B - B_1$ and the return value is less than $\frac{1}{2}(B - B_1)$.

Finally, we need to prove that $M_k \ge B - B_1$ for all $k$ if such a $k$ does not exist.
We have already verified this to be true for $k < n$.
Each matching for $k\ge n$ can be represented as swapping sequences $a$ and $c$ inside the
necklace $x$, developed via an index shift of $k-n$.
The two halves of the necklace $x$ are analogous; thus, all prior observations
of the matching structure remain valid.

If the answer to \maxconvlower for sequences $a,b,c$ is YES, then
$\forall_{k \in [0,n)} \exists_{i+j=k} a[i] + b[j] \ge c[k]$.
Property (1) guarantees that $a\le c$; thus, we conclude that
$\forall_{k \in [0,n)} \exists_{i+j=k} c[i] + b[j] \ge a[i] + b[j] \ge c[k] \ge a[k]$,
and by the same argument as before, the cost of the solution is at least $\frac{1}{2}(B - B_1)$.
\end{proof}

Observe that both \necklace and \maxconvlower admit simple linear
nondeterministic algorithms.
For \maxconvlower, it is sufficient to either assign each $k$
a single condition $a[i] + b[k-i] \ge c[k]$ that is satisfied
or to nondeterministically guess a value of $k$ for which no inequality holds.
For \necklace, we define a decision version of the problem by asking if there
is an alignment of the value bounded by $K$ (the problem is self-reducible via binary search).
For positive instances, the algorithm simply nondeterministically guesses $k$, inducing an optimal solution.
For negative instances, $M_k > 2 K$ must hold for all $k$.
Therefore, it suffices to nondeterministically guess for each $k$ a pair $i,j$ such that
$d\big(x[i], y[(k+i)\; (\text{mod}\; n)]\big) - d\big(x[j], y[(k+j)\; (\text{mod}\; n)]\big) > 2 K$.

In Section~\ref{sec:nondet}, we will show that \maxconvupper
admits an $\Oh\left(n^{1.5}\mathrm{polylog}(n)\right)$
nondeterministic algorithm (see Lemma~\ref{cor:maxconv-nondet})
so, in fact, there is no obstacle to the existence of a subquadratic reduction from \maxconvlower to \maxconvupper.
However, the nondeterministic algorithm for \textsc{3sum} exploits techniques significantly different from ours,
including modular arithmetic. A potential reduction would probably need to rely on
some different structural properties of \maxconv.

\section{Nondeterministic algorithms}
\label{sec:nondet}

Recently, Abboud et al.~\cite{bicriteria} proved that the running time for
the \textsc{Subset Sum} problem cannot be improved to $\Oh(t^{1-\eps}2^{o(n)})$, assuming the SETH.
It is tempting to look for an analogous lower bound for \textsc{Knapsack} that
would make the $\Oh(nt)$-time algorithm tight.
In this section, we take advantage of the nondeterministic lens introduced by Carmosino et al.~\cite{nondet-seth} to argue that the existence of this lower bound for \uknapsack is unlikely.

We recall that by a time complexity of a nondeterministic algorithm, we refer to a bound on running times for
both nondeterministic and co-nondeterministic routines determining whether an instance belongs to the language.
Assuming the Nondeterministic Strong Exponential Time Hypothesis (NSETH), we cannot break the $\Oh(2^{(1-\varepsilon)n})$ barrier for SAT
even with nondeterministic algorithms.

The informal reason to rely on the NSETH is that if we decide to base lower bounds on the SETH, then we should believe that SAT is indeed a very hard problem that does not admit any hidden structure that has eluded researchers so far.
On the other hand, the NSETH can be used to rule out deterministic reductions from SAT to problems with nontrivial nondeterministic algorithms.
This allows us to argue that in some situations basing a hardness theory on the SETH can be a bad idea.
Moreover, disproving the NSETH would imply nontrivial lower bounds on circuit sizes for $\mathbf{E^{NP}}$~\cite{nondet-seth}.

We present a nondeterministic algorithm for the decision version of \uknapsack with running time $\Oh(t\sqrt{n}\log^3(W))$, where $W$ is the target value.
This means that a running time $\Oh(n^{1-\varepsilon}t)$ for \uknapsack cannot be
ruled out with a deterministic reduction from SAT, under the assumption of the NSETH
(for small $\varepsilon < \frac{1}{2}$).

We begin with an observation that a nontrivial 
nondeterministic algorithm for \textsc{3sum}
entails a~similar result for \maxconvupper.

\begin{lemma}\label{cor:maxconv-nondet}
\maxconvupper admits a nondeterministic $\Oh\left(n^{1.5}\mathrm{polylog}(n)\right)$-time algorithm.
\end{lemma}
\begin{proof}
By combining Theorem \ref{thm:3sum} (which involves a reduction from \maxconvupper to
\convsum), the deterministic (i.e., nonrandomized) reduction from \convsum to \textsc{3sum}~\cite{3sum1},
and the nondeterministic $\Oh\left(n^{1.5}\mathrm{polylog}(n)\right)$-time
algorithm for \textsc{3sum} from~\cite[Lemma 5.8]{nondet-seth},
we obtain an analogous algorithm for \maxconvupper.
\end{proof}

In the next step, we require a more careful complexity analysis of the
nondeterministic algorithm for \textsc{3um} developed by Carmosino et al.~\cite[Lemma
5.8]{nondet-seth}.
Essentially, we claim that the running time can be bounded by $\Oh(\sqrt{n_1n_2n_3}\log^2(W))$, where
$n_1,n_2,n_3$ are sizes of the input sets.
This is just a reformulation of the original proof,
where an $\Oh(n^{1.5})$ nondeterministic time algorithm is given, which we have
presented in the Appendix~\ref{sec:3sum-nondet} for completeness.

In the decision version of \uknapsack, we are additionally given a threshold $W$, and we need to determine whether there is a multiset of items with a total weight of at most $t$ and a total value of at least $W$.

\begin{theorem}
\label{thm:fast-nondet-uknapsack}
The decision version of \uknapsack admits an $\Oh(t\sqrt{n}\log^3(W))$ nondeterministic algorithm.
\end{theorem}
\begin{proof}
We can assume that $n \le t$.
If we are given a YES-instance, then we can just nondeterministically guess the solution
and verify it in $\Oh(t)$ time.

To show that an instance admits no solution, we nondeterministically guess a proof involving an array
$(a[k])_{k=0}^t$ such that $a[k]$ is an upper bound for the total value of items with weights summing to at most $k$.
To verify the proof, we need to check that $a[0] = 0$, $a[t] < W$, $a$ is nondecreasing, and, for each $k$ and each item $(w_i,v_i)$, $a[k] + v_i \le a[k + w_i]$ holds.
Let $(b[k])_{k=0}^t$ be a sequence defined as follows:
if there is an item with $w_i = k$, then we set $b[k] = v_i$ (if there are
multiple items with the same weight, we choose the most valuable one) and 0 otherwise.
The latter condition is equivalent to determining if $a \oplus^{\max} b \le a$, which is an instance of \textsc{MaxConv UpperBound}
with elements bounded by $W$.

Note that the sequence $b$ contains only $n$ nonzero elements.
After we have verified (in $\Oh(t)$ time) that $a$ is nondecreasing, we know that $b[j] = 0$
implies $a[i] + b[j] \le a[i+j]$.
This means that we can neglect the zero values in sequence $b$ when applying the reduction in Theorem~\ref{thm:3sum}.
After performing the reduction, we obtain $\Oh(\log W)$ instances of $\textsc{3sumConv}$ with sequences $x, y, z$ of length $t$
but with the additional knowledge that there are only $n$ indices $j$ such that $x[i] + y[j] > z[i+j]$ might hold. 
In the end, we perform a deterministic reduction from \textsc{3sumConv} to \textsc{3sum} in time $\Oh(t)$~\cite{3sum1}.
Since we can omit all but $n$ indices in sequence $y$, we obtain $\Oh(\log W)$ instances of \textsc{3sum}
with set sizes of $t,n$, and $t$.
The claim follows by applying Lemma~\ref{lem:3sum-nondet} and the fact that
nonrandomized reductions preserve the nondeterministic running time.
\end{proof}


From~\cite[Corollary 5.2]{nondet-seth} and the nondeterministic algorithm from
Theorem~\ref{thm:fast-nondet-uknapsack}, it follows that the reduction from any
SETH-hard problem to \uknapsack is unlikely:

\begin{corollary}
    Under the NSETH, there is no deterministic (fine-grained) reduction from the SETH to
    solving \uknapsack in time $\Oh(n^{0.5+\gamma} \cdot t)$ for any $\gamma > 0$.
\end{corollary}

For a natural (but rather technical) definition of \emph{fine-grained reduction}, see~\cite[Definition 3.1]{nondet-seth}.

\section{Conclusions and future work}

In this paper, we undertake a systematic study of \minconv as 
a hardness assumption and prove the subquadratic equivalence of \minconv with \superadditivity, \uknapsack, \knapsack, and \sparsity.
The \minconv conjecture 
is stronger than the well-known conjectures APSP and 3SUM.
Proving that \minconv is equivalent to either APSP or 3SUM would solve a long-standing open problem.
An intriguing question is to determine whether the \minconv conjecture is also stronger than OV.

By exploiting the fast $\Oh(n^2/2^{\Omega(\log n)^{1/2}})$ algorithm for \maxconv, we automatically
obtain $o(n^2)$-time algorithms for all problems in the class. This gives us the
first (to the best of our knowledge) subquadratic algorithm for \superadditivity and improves exact algorithms for
\sparsity by a polylogarithmic factor (although this does not lie within the scope
of this paper).

One consequence of our results is a new lower bound on \knapsack. 
It is known that an $\Oh(t^{1-\epsilon}n^{\Oh(1)})$ algorithm for \knapsack contradicts
the \textsc{SetCover} conjecture~\cite{hard-as-cnf-sat}. 
Here, we show that an $\Oh((n+t)^{2-\epsilon})$ algorithm contradicts the \minconv conjecture. 
This does not rule out an $\Oh(t+n^{\Oh(1)})$ algorithm, which leads to another interesting open problem.

Recently,~\naivecitemany{Abboud}{bicriteria} replaced the \textsc{SetCover}
conjecture with the SETH for \subsetsum.
We have shown that one cannot exploit the SETH to prove that the $\Oh(nt)$-time algorithm 
for \uknapsack is tight.
The analogous question regarding \knapsack remains open.

Finally, it is open whether \maxconvlower is equivalent to \minconv, which would
imply an equivalence between \necklace and \minconv.

\section*{Acknowledgments}

This work is part of a project TOTAL that has received funding from the European
Research Council (ERC) under the European Union’s Horizon 2020 research and
innovation programme (grant agreement no. 677651). We would like to thank Amir
Abboud, Karl Bringmann, Marvin K{\"{u}}nnemann, Oliver Serang and Virginia Vassilevska Williams for the helpful
discussions. The authors would like to thank the anonymous reviewers for their
valuable comments and suggestions.

\bibliographystyle{plainnat}
\bibliography{bib}

\appendix

\section{Reduction to 3SUM}
\label{sec:3sum}

In this section, we show a connection between \maxconv and the \textsc{3sum} conjecture. This
reduction is widely known in the community but, to the best of our knowledge, has
never been explicitly written. We include it in this appendix for completeness.

In this paper, we showed an equivalence between \maxconv and \maxconvupper (see
Theorem~\ref{thm:main}). Also, it is known that the \convsum problem is subquadratically
equivalent to \textsc{3sum}~\cite{3sum1}. Hence, the following theorem suffices.

\begin{theorem}[\maxconvupper $\rightarrow$ \convsum]
\label{thm:3sum}
If \convsum can be solved in time $T(n)$,
then \maxconvupper admits an algorithm with running time $\Oh\left(T(n)\right)$.
\end{theorem}
The proof heavily utilizes \cite[Proposition 3.4, Theorem
3.3]{weighted-subgraphs}, which we present here for completeness. 
$\mathrm{pre}_i(x)$ denotes the binary prefix of $x$ of length $i$,
where the most significant bit is considered the first. 
In the original statement (Proposition 3.4~\cite{weighted-subgraphs}), the prefixes are alternately treated as integers or strings.
We modify the notation slightly to work only with integers.

\begin{lemma}[Proposition 3.4 \cite{weighted-subgraphs}]
\label{lem:pre}
For three integers $x, y, z$, we have that $x + y > z$ iff one of the following holds:
\begin{enumerate}
\item there exists a $k$ such that $\mathrm{pre}_k(x) + \mathrm{pre}_k(y) = \mathrm{pre}_k(z) + 1$,
\item there exists a $k$ such that
\begin{eqnarray}
\mathrm{pre}_{k+1}(x) &= &2\cdot\mathrm{pre}_k(x) + 1, \label{line:pre1} \\
\mathrm{pre}_{k+1}(y) &= &2\cdot\mathrm{pre}_k(y) + 1, \\
\mathrm{pre}_{k+1}(z) &= &2\cdot\mathrm{pre}_k(z), \label{line:pre3} \\
\mathrm{pre}_k(z) & = & \mathrm{pre}_k(x) + \mathrm{pre}_k(y). \label{line:pre4}
\end{eqnarray}
\end{enumerate}
\end{lemma}

\begin{proof}[Proof of Theorem \ref{thm:3sum}]
We translate the inequality $a[i] + b[j] > c[i+j]$ from \maxconvupper
to an alternative of $2\log{W}$ equations.
For each $0 \le k \le \log{W}$, we construct two instances of \convsum related
to the conditions in Lemma~\ref{lem:pre}.
For the first condition, we create sequences $a^k[j] = \mathrm{pre}_k(a[j]),\, b^k[j] = \mathrm{pre}_k(b[j]),\, c^k[j] = \mathrm{pre}_k(c[j]) + 1$.
For the second one, we choose a value of $D$ that is two times larger than the absolute value of any element and set

\begin{eqnarray*}
\tilde{a}^k[j] &=
  &\begin{cases}
    \mathrm{pre}_k(a[j])       & \quad \text{if } \mathrm{pre}_{k+1}(a[j]) = 2\cdot\mathrm{pre}_k(a[j]) + 1, \\
    -D  & \quad \text{otherwise}, \\
  \end{cases} \\ 
\tilde{b}^k[j] &=
  &\begin{cases}
    \mathrm{pre}_k(b[j])       & \quad \text{if } \mathrm{pre}_{k+1}(b[j]) = 2\cdot\mathrm{pre}_k(b[j]) + 1, \\
    -D  & \quad \text{otherwise}, \\
  \end{cases} \\
\tilde{c}^k[j] &=
  &\begin{cases}
    \mathrm{pre}_k(c[j])       & \quad \text{if } \mathrm{pre}_{k+1}(c[j]) = 2\cdot\mathrm{pre}_k(c[j]), \\
    D  & \quad \text{otherwise}. \\
  \end{cases} 
\end{eqnarray*}
Observe that if any of the conditions \ref{line:pre1} -- \ref{line:pre3} is not satisfied, then the unrolled
formula $\tilde{a}^k[i] + \tilde{b}^k[j] = \tilde{c}^k[i+j]$ contains
at least one summand $D$ and thus cannot be satisfied.
Otherwise, it reduces to the condition~\ref{line:pre4}.

The inequality $a[i] + b[j] > c[i+j]$ holds for some $i,j$ iff
one of the constructed instances of \convsum returns \emph{true}.
As the number of instances is $\Oh(\log{W})$, the claim follows.
The \convsum problem is subquadratically equivalent to \textsc{3sum}~\cite{3sum1},
which establishes a relationship between these two classes of subquadratic equivalence.
\end{proof}

\section{Nondeterministic algorithm for 3SUM}
\label{sec:3sum-nondet}

Carmosino et al.~\cite[Lemma 5.8]{nondet-seth} presented an $\Oh(n^{1.5})$
nondeterministic algorithm for \textsc{3sum}, i.e., the running time depends
only on the size of the input. However, in our application, we need a running time
that is a function of the sizes of sets $A$,$B$ and $C$. In this section we
analyze the running time of Carmosino et al. in regard to these parameters.

\begin{lemma}\label{lem:3sum-nondet}
There is a nondeterministic algorithm for \textsc{3sum} with
running time \\ $\Oh(\sqrt{n_1n_2n_3}\log^2(W))$, where
$n_1 = |A|,\, n_2 = |B|,\, n_3 = |C|$ and $W$ is the maximum absolute value of integers in $A \cup B \cup C$ (we assume that $n_1 + n_2 + n_2 \le W$).
\end{lemma}
\begin{proof}
If there is a triple $(a\in A,\,b\in B,\,c\in C)$ such that $a+b=c$, then we can
    nondeterministically guess it and verify it in $\Oh(1)$ time. To prove that there is no such triple,
we nondeterministically guess the following:
\begin{enumerate}
\item a prime number $p \le prime_{\sqrt{n_1n_2n_3}}$, where $prime_i$ denotes the $i$-th prime number,
\item an integer $t(p) \le \sqrt{n_1n_2n_3}\log(3W)$, which is the number of solutions for sets \\ $(A \Mod p,\, B \Mod p,\, C \Mod p)$,
\item a set $S = \{(a_1,b_1,c_1),\dots,(a_{t(p)},b_{t(p)},c_{t(p)})\}$, where $|S| = t(p)$ and each triple $(a_i\in A,\,b_i\in B,\,c_i\in C)$ satisfies $a_1 + b_1 \equiv c_1 (\mathrm{mod}\, p)$.
\end{enumerate}

To see that for each NO-instance there exists such a proof, consider the number of false positives, that is, tuples $(a\in A,\,b\in B,\,c\in C,\,p)$, where $p$ is a prime.
For each triple $(a\in A,\,b\in B,\,c\in C)$, the value $|a+b-c|$ has at most $\log(3W)$ distinct prime divisors. Therefore, the number of false positives
is bounded by $n_1n_2n_3\log(3W)$.
Since there are $\sqrt{n_1n_2n_3}$ candidates for $p$, we can choose one such that $t(p) \le \sqrt{n_1n_2n_3}\log(3W)$.

To verify the proof, we need to verify whether $S$ contains no true solution and to compute the number of solutions for $(A \Mod p,\, B \Mod p,\, C \Mod p)$.
If it equals to $|S|$, then we are sure that all solutions for the instance modulo $p$ are indeed false positives for the original instance.
Since the numbers are bounded by $p$, we can count the solutions using FFT in time $\Oh(p\log p) = \Oh(\sqrt{n_1n_2n_3}\log^2(W))$. 
\end{proof}

\end{document}